\documentclass[10pt,conference,letterpaper]{IEEEtran}

\usepackage[OT1]{fontenc}
\usepackage{cite}
\usepackage{latexsym}
\usepackage{amsmath}
\usepackage{amsthm}
\usepackage{amssymb}
\usepackage{mathrsfs}

\usepackage{subfigure}
\usepackage{algorithm}
\usepackage{algorithmic}
\usepackage{multirow}
\usepackage{bbm}
\usepackage{arydshln}
\usepackage{url}
\usepackage{pict2e,picture} 
\usepackage{tikz}   
\usepackage{pifont} 
\usepackage{enumitem}
\usepackage[bottom]{footmisc}

\newcommand{\f}{\mathbf{f}}

\newcommand{\I}{\mathbf{I}}

\newcommand{\W}{\mathbf{W}}

\newcommand{\x}{\mathbf{x}}

\newcommand{\y}{\mathbf{y}}
\newcommand{\z}{\mathbf{z}}
\newcommand{\Z}{\mathbf{Z}}

\newcommand{\1}{\mathbf{1}}

\newtheorem{thm}{Theorem}

\newtheorem{lem}{Lemma}

\newtheorem{prop}[thm]{Proposition}
\newtheorem{ex}{Example}
\newtheorem{defn}{Definition}

\newtheorem{rem}{Remark}

\newtheorem{assum}{Assumption}

\begin{document}

\title{Compressed Distributed Gradient Descent: Communication-Efficient Consensus over Networks}



\author{Xin Zhang$^{\dag}$ \mbox{\hspace{0.4cm}} Jia Liu$^{\ddag}$ \mbox{\hspace{0.4cm}} Zhengyuan Zhu$^{\dag}$ \mbox{\hspace{0.4cm}} Elizabeth S. Bentley$^{*}$
\\ $^{\dag}$Department of Statistics, Iowa State University
\\ $^{\ddag}$Department of Computer Science, Iowa State University
\\ $^{*}$Air Force Research Laboratory, Information Directorate
\thanks{
{\color{blue} 
This work has been supported by NSF grants CNS-1527078, 1514260, 1446582, 1409336, 1012700, WiFiUS-1456806, ECCS-1444026, 1232118; ONR grant N00014-15-1-2166; ARO grant W911NF-14-1-0368; DTRA grants HDTRA 1-14-1-0058, 1-15-1-0003, AFRL VFRP'15 award; DARPA grant HROOll-15-C-0097and QNRF grant NPRP 7-923-2-344.
DISTRIBUTION STATEMENT A: Approved for Public Release; distribution unlimited 88ABW-2015-5989 on Dec. 14, 2015. 
}
}
}

\maketitle

\begin{abstract}
Network consensus optimization has received increasing attention in recent years and has found important applications in many scientific and engineering fields.
To solve network consensus optimization problems, one of the most well-known approaches is the distributed gradient descent method (DGD).
However, in networks with slow communication rates, DGD's performance is unsatisfactory for solving high-dimensional network consensus problems due to the communication bottleneck.
This motivates us to design a communication-efficient DGD-type algorithm based on compressed information exchanges.
Our contributions in this paper are three-fold:
i) We develop a communication-efficient algorithm called amplified-differential compression DGD (ADC-DGD) and show that it converges under {\em any} unbiased compression operator;
ii) We rigorously prove the convergence performances of ADC-DGD and show that they match with those of DGD without compression;
iii) We reveal an interesting phase transition phenomenon in the convergence speed of ADC-DGD.
Collectively, our findings advance the state-of-the-art of network consensus optimization theory.
\end{abstract}

\section{Introduction} \label{sec:intro}

In recent years, network consensus optimization has received increasing attention  thanks to its generality and wide applicability.
To date, network consensus optimization has found important applications in many  scientific and engineering fields, e.g., distributed sensing in wireless sensor networks\cite{ling2010decentralized,predd2005distributed,schizas2008consensus,zhao2002information}, decentralized machine learning\cite{duchi2012dual,tsianos2012consensus}, multi-agent robotic systems\cite{cao2013overview,ren2007information,zhou2011multirobot}, smart grids\cite{giannakis2013monitoring,kekatos2013distributed}, to name just a few.
Simply speaking, in a network consensus optimization problem, each node only has access to some component of the global objective function.
That is, the global objective function is only partially known at each node.
Through communications with local neighbors, all nodes in the network collaborate with each other and try to reach a consensus on an optimal solution, which minimizes the global objective function.

Among various algorithms for solving network consensus optimization problems, one of the most effective methods is the distributed gradient descent (DGD) algorithm, a first-order iterative method developed by Nedic and Ozdaglar approximately a decade ago\cite{nedic2009distributed}.
The enduring popularity of DGD is primarily due to its implementation simplicity and elegant networking interpretation: 
In each iteration of DGD, each node performs an update by using a linear combination of a gradient step with respect to its local objective function
and a weighted average from its local neighbors (also termed as a consensus step).
It has been shown that DGD enjoys the same $O(1/k)$ convergence speed as the classical gradient descent method, where $k$ denotes the number of iterations\cite{nedic2009distributed}.
The simplicity and salient features of DGD have further inspired a large number of extensions to various network settings (see Section~\ref{sec:related_work} for more in-depth discussions).

However, despite its theoretical and engineering appeals, the performance of DGD may not always be satisfactory in practice.
This is particularly true for solving a high-dimensional consensus problem over a network with low network communication speed.
In this case, due to the large amount of data sharing and the communication bottleneck, exchanging full high-dimensional information between neighboring nodes is time-consuming (or even infeasible), which significantly hinders the overall convergence of DGD.
To improve the convergence speed, several second-order approaches using Hessian approximation (with respect to local objective function) have been proposed (see, e.g., \cite{jordan2016communication,wang2016efficient}).
Although these second-order methods converge in a fewer number of iterations (hence less information exchanges), they require matrix inversion in each iteration, implying  a $\Omega(d^{2}\log d)$ per-iteration complexity for a $d$-dimensional problem.
Hence, for high-dimensional consensus problems (i.e., large $d$), low-complexity first-order methods remain more preferable in practice.

To address DGD's limitations in high-dimensional network consensus over low-speed networks, a naturally emerging idea is to {\em compress} the information exchanged between nodes.
Specifically, by compressing the information in a high-dimensional state space to a smaller set of quantized states, each node can use a codebook  to represent the quantized states with a small number of bits.
Then, rather than directly transmitting full information, each node can just transmit the small-size codewords, which significantly reduces the communication burden.
Moreover, from a cybersecurity standpoint, transmitting compressed information is also very helpful because each node can encrypt its codebook and avoid revealing full information to potential eavesdroppers in the network.

However, with compressed information being adopted in DGD, several fundamental questions immediately arise:
{\em i) Will DGD with compressed information exchanges still converge?
ii) If the answer to i) is no, could we modify DGD to make it work with compressed information?
iii) If the answer to ii) is yes, how fast does this modified DGD method converge?
}
Indeed, answering all these questions are highly non-trivial and they constitute the main subjects of this paper.
The main contribution in this paper is that we provide concrete answers to all three fundamental questions.
Our key results and their significance are summarized as follows:



\begin{list}{\labelitemi}{\leftmargin=1em \itemindent=-0.5em \itemsep=.2em}
\item First, we show that DGD with straightforward compressed information exchange fails to converge because of a non-vanishing accumulated noise term resulted from compression over iterations.
This motivates us to develop a noise variance reduction method.
To this end, we propose a new idea called ``amplified-differential compression DGD'' (ADC-DGD), where, instead of directly exchanging compressed estimates of the global optimization variable in DGD, we exchange an {\em amplified} version of the {\em state differential} between consecutive iterations, hence the name.
We show that ADC-DGD effectively diminishes the accumulated noise from compression and induces convergence. 


\item 
We show that, under {\em any} unbiased compression operator, our ADC-DGD method converges at rate $O(1/k)$ to an $O(\alpha^{2})$-neighborhood of an optimal solution with a constant step-size $\alpha$.
Under diminishing step-sizes, ADC-DGD converges asymptotically at rate $O(1/\sqrt{k})$ to an optimal solution.
We note that these convergence rates are the {\em best possible} in the sense that they match with those of the original DGD without compression.
This result is surprising since the information loss due to compression could be large.
We also note that the convergence rate of ADC-DGD outperforms other existing distributed first-order methods with compression (see Section~\ref{sec:related_work} for detailed discussions).


\item Based on the above convergence results of ADC-DGD, we further investigate the impacts of ADC-DGD's amplifying factor on convergence speed and communication load.
Interestingly, we reveal a phase transition phenomenon of the convergence speed with respect to the amplification exponent $\gamma$ in ADC-DGD.
Specifically, when $\gamma \in (\frac{1}{2}, 1]$ (sublinear growth of amplification), convergence speed approaches that of DGD as $\gamma$ increases.
However, as soon as $\gamma >1$, there is no further convergence speed improvement but network communication load continues to grow.
This shows that $\gamma=1$ is a critical point, under which we can trade communication overhead for convergence speed.
\end{list}

Collectively, our results contribute to a growing theoretical foundation of network consensus optimization.
The rest of the paper is organized as follows. 
In Section~\ref{sec:related_work}, we review related work. 
In Section~\ref{sec:method}, we introduce the network consensus optimization problem and show that DGD with compressed information exchange fails to converge.
In Section~\ref{sec:adc_dgd}, we present our ADC-DGD algorithm and its convergence performance analysis. 
Numerical results are provided in Section~\ref{sec:numerical} and Section~\ref{sec:conclusion} concludes this paper.  

\section{Related Work} \label{sec:related_work}

In this section, we first provide a quick overview on the historical development of DGD-type algorithms.
We then focus on the recent advances of communication-conscious  network consensus optimization, including related work that utilize compression.

\smallskip
{\em 1) DGD-Based Algorithms for Network Consensus:}
Network consensus optimization can trace its roots to the seminal work by Tsitsiklis \cite{tsitsiklis1984problems}, where the system model and the analysis framework were first developed.
As mentioned earlier, a well-known method for solving network consensus optimization is the distributed (sub)gradient descent (DGD) method, which was proposed by Nedic and Ozdaglar in \cite{nedic2009distributed}.
DGD was recently reexamined in \cite{yuan2016convergence} by Yuan {\em et al.} using a new Lyapunov technique, which offers further mathematical understanding of its convergence performance.
In their follow-up work \cite{zeng2016nonconvex}, the convergence behavior of DGD was further analyzed for non-convex problems. 
Recently, several DGD variants have been proposed to enhance the convergence performance (e.g., achieving the same $O(1/k)$ convergence rate with constant step-size \cite{shi2015extra} or even under time-varying network graphs \cite{nedi2017achieving}).

{\em 2) Communication-Conscious Distributed Optimization:}
As mentioned earlier, studies have shown that communication costs of DGD could be a major concern in practice.
To this end, Chow {\em et al.} \cite{chow2016expander} studied the tradeoff between communication requirements and prescribed accuracy.
In \cite{berahas2017balancing}, Berahas {\em et al.} developed an adaptive DGD framework called $\text{DGD}^t$ to balance the costs between communication and computation.
Here, the parameter $t$ represents the number of consensus steps performed per gradient descent step ($t=1$ corresponding to the original DGD).
The larger the $t$-value, the cheaper the communication cost, and vice versa.
The most related works to ours are by Reisizadeh {\em et al.}\cite{reisizadeh2018quantized} and Tang {\em et al.}\cite{tang2018decentralization}, which also consider adopting compression in DGD.
The algorithm in \cite{reisizadeh2018quantized}, focusing on the strongly convex problems, used the diminishing step-size strategy to guarantee the convergence. In our work, we extend the strongly convexity assumption and prove the convengence both on diminishing step-size and constant step-size.
However, our algorithm differs from \cite{tang2018decentralization} in the following key aspects: 
i) The compression in \cite{tang2018decentralization} uses a quantized extrapolation between two successive iterates, which can be viewed as a diminishing step-size strategy.
In contrast, our ADC-DGD algorithm uses an amplified differential of two successive  iterates.
As will be shown later, our algorithm can be interpreted as a {\em variance reduction} method; 
ii) Our convergence rate outperforms that of \cite{tang2018decentralization}. 
The fastest convergence rate of the algorithms in \cite{tang2018decentralization} is $O(\log(k)/\sqrt{k})$, 
while the convergence rate of our ADC-DGD algorithm is $o(1/\sqrt{k})$;
iii) To reach the best convergence rate in \cite{tang2018decentralization}, the extrapolation compression algorithm needs to solve a complex equation to obtain an optimal step-size. 
In contrast, our ADC-DGD algorithm uses the standard sublinearly diminishing step-sizes, which is of much {\em lower} complexity and can be easily implemented in practice.

\section{Network Consensus Optimization and Distributed Gradient Descent} \label{sec:method}

In Section~\ref{subsec:consensus}, we first introduce the network consensus optimization problem, which is followed by the basic version of the DGD method.
Then in Section~\ref{subsec:motivating_example}, we will illustrate an example where DGD with directly compressed information fails to converge, which motivates our subsequent ADC-DGD approach in Section~\ref{sec:adc_dgd}.

\subsection{Consensus Optimization over Networks: A Primer} \label{subsec:consensus}
Consider an undirected connected graph $\mathcal{G}=(\mathcal{N,\mathcal{L}})$,
where $\mathcal{N}$ and $\mathcal{L}$ are the sets of nodes and links, respectively, with $|\mathcal{N}| = N$ and $|\mathcal{L}| = E$. 
Let $x \in \mathbb{R}^{P}$ be some global decision variable to be optimized.
Each node $i$ has a local objective function $f_i(x),$ $i=1,\cdots, N$ (only available to node $i$).
The global objective function is the sum of all local objectives, i.e., $f(x) \triangleq \sum_{i=1}^{N} f_{i}(x)$.
Our goal is to solve the following network-wide optimization problem in a {\em distributed} fashion: 
\begin{align} \label{eqn_general_problem}
\min_{x \in \mathbb{R}^{P}} f(x) = \min_{x \in \mathbb{R}^{P}} \sum_{i=1}^{N} f_{i}(x).
\end{align}
Problem~(\ref{eqn_general_problem}) has a wide range of applications in practice.
For example, consider a wireless sensor network, where each sensor node $k$ distributively collects some local monitored temporal data $\x^k=(x_1^k,\cdots,x_T^k)$ and collaborates to detect the change-point in the global temporal data. 
This problem can be formulated as: $\min \sum_{i=1}^{N} f_{i}(\x^{k})$,
where $f_{i}(\x^{k}) \triangleq -|\sum_{i=1}^{t} x_i^k - \frac{t}{T}\sum_{i=1}^{T} x_i^k|^2$ is the  CUSUM (cumulative sum control chart) statistics. 
Note that Problem~(\ref{eqn_general_problem}) can be equivalently written in the following {\em consensus} form:
\vspace{-.05in}
\begin{align} \label{Eq:problem1}
& \text{Minimize} && \hspace{-.5in} \sum_{i=1}^{n} f_i(x_i) & \\
& \text{subject to} && \hspace{-.5in} x_i=x_j, && \hspace{-.5in} \forall (i,j) \in \mathcal{L}. \nonumber
\vspace{-.05in}
\end{align}
where $x_i\in \mathbb{R}^P$ is the local copy of $x$ at node $i$. 
In Problem (\ref{Eq:problem1}), the constraints enforce that the local copy at each node is equal to those of its neighbors, hence the name consensus.
It is well-known \cite{nedic2009distributed} that Problem (\ref{Eq:problem1}) can be reformulated as:
\vspace{-.03in}
\begin{align} \label{Eq:problem2}
& \text{Minimize} && \hspace{-.5in} \sum_{i=1}^{n} f_i(x_i) & \\
& \text{subject to} && \hspace{-.5in} (\W \otimes \I_{P}) \x = \x, &&\nonumber
\vspace{-.09in}
\end{align}
where $\x \!\triangleq\! [x_1^{\top},\ldots,x_n^{\top}]^{\top} \!\in\! \mathbb{R}^{NP}$, $\I_{P}$ denotes the $P$-dimensional identity matrix,
and the operator $\otimes$ denotes the Kronecker product.
In (\ref{Eq:problem2}), $\mathbf{W} \in \mathbb{R}^{N\times N}$ is referred to as the consensus matrix and satisfies the following properties:
\begin{enumerate}
\item $\mathbf{W}$ is doubly stochastic: $\sum_{i=1}^{N} [\mathbf{W}]_{ij}=\sum_{j=1}^{N} [\mathbf{W}]_{ij}=1$.
\item The sparsity pattern of $\W$ follows the network topology: $[\W]_{ij} > 0$ for $\forall~(i,j)\in \mathcal{L}$ and $[\mathbf{W}]_{ij}=0$ otherwise. 
\item $\mathbf{W}$ is symmetric and hence it has real eigenvalues. 
\end{enumerate}
The doubly stochastic property in 1) ensures that all eigenvalues of $\mathbf{W}$ are in $(-1,1]$ and exactly one eigenvalue is equal to 1.
Hence, it follows from Property 3) that one can sort eigenvalues as $1=\lambda_1(\mathbf{W})\ge\cdots\ge\lambda_N(\mathbf{W})> -1.$ 
Let $\beta \triangleq \max\{|\lambda_2(\mathbf{W})|,|\lambda_N(\mathbf{W})|\}$.
Clearly, we have $\beta<1.$
It is shown in \cite{nedic2009distributed} that $(\W \otimes \I_{P}) \x = \x$ if and only if $x_i = x_j$, $(i,j) \in \mathcal{L}$. 
Therefore, Problems (\ref{Eq:problem1}) and (\ref{Eq:problem2}) are equivalent.


The equivalent network consensus formulation in Problem~(\ref{Eq:problem2}) motivates the design of the decentralized gradient descent (DGD) method as stated in Algorithm~1:


\medskip
\hrule 
\vspace{.03in}
\noindent {\bf Algorithm~1:} Decentralized Gradient Descent (DGD)\cite{nedic2009distributed}.
\vspace{.03in}
\hrule
\vspace{0.1in}
\noindent {\bf Initialization:}
\begin{enumerate} [topsep=1pt, itemsep=-.1ex, leftmargin=.2in]
\item[1.] Let $k=1$. Choose initial values for $x_{i,1}$ and step-size $\alpha_1$. 
\end{enumerate}

\noindent {\bf Main Loop:}
\begin{enumerate} [topsep=1pt, itemsep=-.1ex, leftmargin=.2in]
\item[2.] In the $k$-th iteration, each node sends its local copy to its neighbors. Also, upon reception of all local copies from its neighbors, each node updates its local copy as follows:
   \begin{equation} \label{eqn:orig_dgd}
   x_{i,k+1} = \underbrace{\sum\nolimits_{j \in \mathcal{N}_{i}} [\mathbf{W}]_{ij}x_{j,k}}_{\mathrm{Consensus \,\, step}} - \underbrace{\alpha_k\nabla f_i(x_{i,k})}_{\mathrm{Gradient \,\, step}},
   \end{equation}
where $[\W]_{ij}$ is the entry in the $i$-th row and $j$-th column in $\W$, $x_{i,k}$ and $\alpha_{k}$ represent $x_{i}$'s value and step-size in the $k$-th iteration, respectively, and $\mathcal{N}_{i} \!\triangleq\! \{ j \!\in\! \mathcal{N}: (i,j) \!\in\! \mathcal{L}\}$.
\item[3.] Stop if a desired convergence criterion is met; otherwise, let $k \leftarrow k+1$ and go to Step 2. 
\end{enumerate}
\smallskip
\hrule
\medskip

We can see that the DGD update in (\ref{eqn:orig_dgd}) consists of a consensus step and a local gradient step, which can be easily implemented in a network.
Also, DGD achieves the same $O(1/k)$ convergence rate as in the classical gradient descent method.
However, as mentioned in Section~\ref{sec:intro}, DGD may not work well for high-dimensional consensus problem in low-speed networks.
Hence, we are interested in developing a DGD-type algorithm with compressed information exchanges in this paper.
In what follows, we will first show that DGD fails to converge if compressed information is directly adopted in the consensus step.

%

\subsection{DGD with Directly Compressed Information Exchange Does Not Converge: A Motivating Example} \label{subsec:motivating_example}

We first introduce the notion of unbiased stochastic compression operator, which has been widely used to represent compressions in the literature (see, e.g., \cite{chow2016expander,zhang2012communication,berahas2017balancing,alistarh2017qsgd,wen2017terngrad,wangni2017gradient}). 
\begin{defn}[Unbiased Stochastic Compression Operator]\label{defn:usco}
{\em A stochastic compression operator $C(\cdot)$ is unbiased if it satisfies $C(z)=z+\epsilon_z$, with $\mathbb{E}[\epsilon_z]=0$ and $\mathbb{E}[\epsilon_z^2]\le \sigma^2$}, $\forall~z$. 
\end{defn}

Defintion~\ref{defn:usco} guarantees that the noise caused by the compression has no effect on the mean of the parameter and its variance is bounded. 
Many compressed operators satisfy the above definition. 
The following are some examples:

\begin{ex}[The Low-precision Quantizer \cite{reisizadeh2018quantized}]
Consider the partition for the real line $\mathbb{R}:$ $\{\cdots, a_{-k},\cdots,a_0,\cdots,a_{k},\cdots \},$ with $a_i<a_{i+1}$ $\forall i.$
For ${\z}=(z_1,\cdots,z_p)^{\top} \in \mathbb{R}^P,$ the $k$-th element of $[C({\z})]_k$ is: if $a_{i}\le z_k<a_{i+1},$
\begin{align*}
[C({\z})]_k=
\begin{cases}
a_i, \text{with probability}~ \frac{a_{i+1}-z_k}{a_{i+1}-a_{i}}, \\
a_{i+1},\text{with probability}~ 1- \frac{a_{i+1}-z_k}{a_{i+1}-a_{i}}.
\end{cases}
\end{align*}
\end{ex}

\begin{ex}[The Randomly Rounding Operater \cite{alistarh2017qsgd}] 
For ${\z}=(z_1,\cdots,z_p)^{\top} \in \mathbb{R}^P,$ the $k$-th element of $[C({\z})]_k$ is:
\begin{align*}
[C({\z})]_k=
\begin{cases}
\lfloor z_k \rfloor+1, \text{with probability}~(1-p_k), \\
\lfloor z_k \rfloor,\text{with probability}~p_k.
\end{cases}
\end{align*}
where $\lfloor z \rfloor$ presents the largest integer smaller than $z$ and the probability $p_k=z_k-\lfloor z_k \rfloor$.
\end{ex}

\begin{ex}[The quantization sparsifier] 
Consider the $m$-partition for $\mathbb{B}(0,M):$ $\{a_0=0, a_1,\cdots,a_{m-1}, a_m =M \},$ with $a_i<a_{i+1}$ $\forall i.$
For a bounded vector ${\z}=(z_1,\cdots,z_p)^{\top}$ with $|z_i|\le M,$ the $k$-th element of $[C({\z})]_k$ is: if $a_{i}\le z_k<a_{i+1},$
\begin{align*}
[C({\z})]_k=
\begin{cases}
 \text{sign}(z_k)\cdot a_{i+1}, \text{with probability}~ \frac{z_k}{a_{i+1}}, \\
 0,\text{with probability}~1-\frac{z_k}{a_{i+1}}.
\end{cases}
\end{align*}
\end{ex}


Now, we consider the convergence of DGD with unbiased stochastic compressions.
If local copies are compressed and then directly used in the consensus step in the DGD algorithm, then Eq.~(\ref{eqn:orig_dgd}) in Algorithm~1 can be modified as:
\begin{align} \label{eqn:DGD_diverge}
&x_{i,k+1} =\sum\nolimits_{j \in \mathcal{N}_{i}} [\mathbf{W}]_{ij}C(x_{j,k})-\alpha_k\nabla f_i(x_{i,k}) = \nonumber\\
& \underbrace{ \sum\nolimits_{j=\mathcal{N}_{i}} [\mathbf{W}]_{ij}x_{j,k} \!-\!\alpha_k\nabla f_i(x_{i,k})}_{\mathrm{Exact \,\, DGD}}+\!\!\!
\underbrace{\sum\nolimits_{j \in \mathcal{N}_{i}} [\mathbf{W}]_{ij}\epsilon_{x_{j,k}}}_{\mathrm{Accumulated \,\, noise \,\, term}}\!\!, \!\!\!\!\!\!
\end{align}
which shows that there is a {\em non-vanishing} noise term accumulated over iterations, which prevents the DGD algorithm from converging. 
For example, consider a simple 2-node network with local objectives $f_1(x)= 4(x-2)^2$ and $f_2(x)= 2(x+3)^2$.
The quantized compressed operator\cite{wen2017terngrad} is adopted in DGD.
The simulation results are illustrated in Fig.~\ref{Fig:countercase}, where we can see that DGD fails to converge after 1000 iterations even for such a small-size network consensus problem.
This motivates us to pursue a new algorithmic design in Section~\ref{sec:adc_dgd}.

\begin{figure}[t!]
\centering
\subfigure[$\alpha = 0.001$.]{%
\label{fig:countercase1}%
\includegraphics[width=0.23\textwidth]{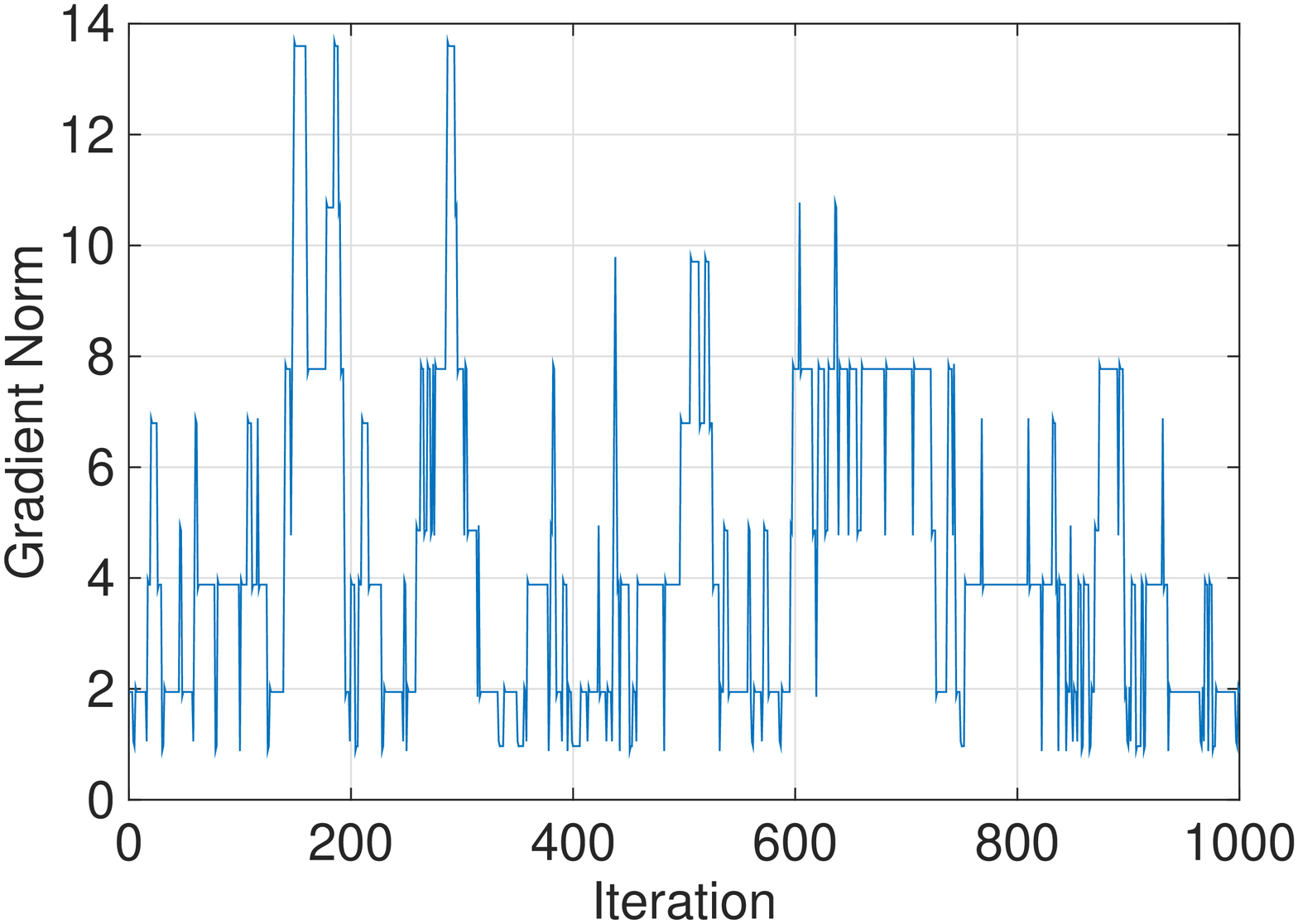}}%
\hspace{8pt}%
\subfigure[$\alpha = 0.001/\sqrt{k}$.]{%
\label{fig:countercase2}%
\includegraphics[width=0.23\textwidth]{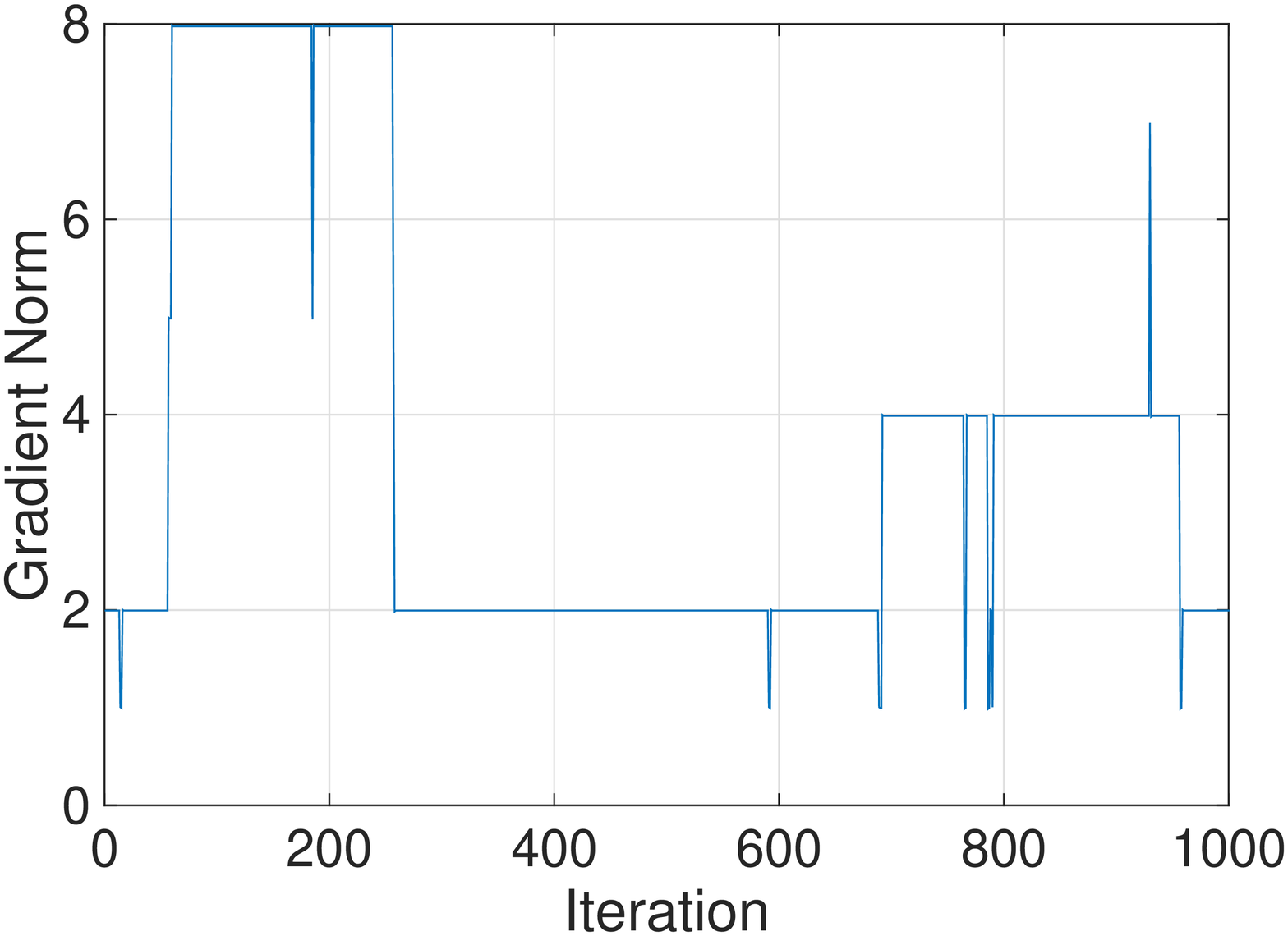}}
\caption{The simulation results for DGD with quantized compression operator for a 2-node network, for which DGD fails to converge after 1000 iterations.}%
\label{Fig:countercase}%
\vspace{-.15in}
\end{figure}

\section{Amplified-Differential Distributed Gradient Descent Method (ADC-DGD)} \label{sec:adc_dgd}

In this Section, we will first introduce our ADC-DGD algorithm in Section~\ref{subsec:the_algorithm}.
Then, we will present the main theoretical results and their intuitions in Section~\ref{subsec:main_results}.
The proofs for the main results are provided in Section~\ref{subsec:conv_analysis}.

\subsection{The ADC-DGD Algorithm} \label{subsec:the_algorithm}

Our ADC-DGD algorithm is stated in Algorithm~2:

%

\medskip
\hrule 
\vspace{.03in}
\noindent {\bf Algorithm~2:} Amplified-Differential Compression DGD.
\vspace{.03in}
\hrule
\vspace{0.1in}
\noindent {\bf Initialization:}
\begin{enumerate} [topsep=1pt, itemsep=-.1ex, leftmargin=.2in]
\item[1.] Let $k\!=\!1$. Let $x_{i,0}\!=\!\tilde{x}_{i,0}\!=\!0$, $\forall i$. 
Choose initial values for step-size $\alpha_1$ and the amplification exponent $\gamma$.
Let $x_{i,1}\!=\!y_{i,1}\!=\!-\alpha_1\nabla f_i(x_{i,0})$, $\forall i$. 
\end{enumerate}

\noindent {\bf Main Loop:}
\begin{enumerate} [topsep=1pt, itemsep=-.1ex, leftmargin=.2in]
\item[2.] In the $k$-th iteration, each node sends the compressed amplified-differential $d_{i,k}=C(k^\gamma y_{i,k})$ to its neighbors. 
Also, upon collecting all neighbors' information, each node estimates neighbors' (imprecise) values: $\tilde{x}_{j,k} \!=\! \tilde{x}_{j,k-1} \!+\! d_{j,k}/k^\gamma$.
Then, each node updates its local value:
   \begin{equation} \label{eqn:adc_dgd}
   x_{i,k+1} = \underbrace{\sum\nolimits_{j \in \mathcal{N}_{i}} [\mathbf{W}]_{ij}\tilde{x}_{j,k}}_{\mathrm{Compressed \,\,consensus}} - \underbrace{\alpha_k\nabla f_i(x_{i,k})}_{\mathrm{gradient \,\, step}}.
   \end{equation}
Each node updates local differential: $y_{i,k+1}\!=\!x_{i,k+1} \!-\! \tilde{x}_{i,k}$.
\item[3.] Stop if a desired convergence criterion is met; otherwise, let $k \leftarrow k+1$ and go to Step 2. 
\end{enumerate}
\smallskip
\hrule
\medskip

Several important remarks on Algorithm~2 are in order:
i) Compared to the original DGD, each node $i$ under ADC-DGD requires additional memory to store the (imprecise) values of its neighbors in the previous iteration: $\{\tilde{x}_{j,k-1}: (i,j)\in \mathcal{L}\}$.
This additional memory allows the neighbors to only transmit the difference between successive iterations $y_{i,k}=x_{i,k}-\tilde{x}_{i,k-1}$ rather than $x_{i,k}$ directly.
Note that this memory requirement is modest in practice since many computer networks are scale-free (i.e., node degree distribution follows a power law and hence most nodes have low degrees);
ii) Each node sends out a compressed version of the amplified-differential $C(k^\gamma y_{i,k})$.
This information will then be de-amplified at the receiving nodes as $d_{i,k}/k^{\gamma}$, which is a noisy version of $y_{i,k}$.
Based on the memory of the previous version, each node obtains their neighbors' values estimation $\tilde{x}_{j,k}$, $j \in \mathcal{N}_{i}$.
Clearly, ADC-DGD is more communication-efficient compared to the original DGD;
iii) Once $\tilde{x}_{j,k}$, $j \in \mathcal{N}_{i}$, are available, the update in (\ref{eqn:adc_dgd}) follows the same structure as in DGD, which also contains a consensus step and a local gradient step.
Therefore, the complexity of ADC-DGD are almost identical to the original DGD, which means that ADC-DGD enjoys the same low-complexity.

%
%

\subsection{Main Convergence Results} \label{subsec:main_results}
Before presenting the convergence results of ADC-DGD, we first state several needed assumptions:
\begin{assum}\label{Assump:Obj}
The local objective functions $f_i(\cdot)$ satisfy:
\begin{list}{\labelitemi}{\leftmargin=1em \itemindent=-0.5em \itemsep=.2em}
\item (Lower boundedness) There exists an optimal $x^*$ with $\|x^*\| \!<\! \infty$ such that $\forall x \!\neq\! x^*,$ $\sum_{i=1}^{N}f_i(x) \!\ge\! \sum_{i=1}^{N}f_i(x^*);$
\item (Lipschitz continuous gradient) there exists a constant $L > 0$ such that $\forall x,y,$ $\|\nabla f_i(x)-\nabla f_i(y)\|\le L\|x-y\|,$ $\forall i$.
\end{list}
\end{assum}

\begin{assum}[Growth rate at infinity]\label{Assump:ObjInfity}
If the domain for $\x$ is unbounded, then there exists a constant $M>0$ such that $$\lim_{\|\x\|\rightarrow \infty}\frac{\|\x\|}{f(\x)}=\lim_{\sum_{i=1}^{n}\|x_i\|\rightarrow \infty}\frac{\sum_{i=1}^{N}\|x_i\|}{\sum_{i=1}^Nf_i(x_i)}\le M,$$
where $f(\x)=\sum_{i=1}^Nf_i(x_i)$ and $\x=(x_1^{\top},\cdots,x_N^{\top})^{\top}.$
 \end{assum}

Assumption~\ref{Assump:Obj} is standard in convergence analysis of gradient descent type algorithms: The first bullet ensures the existence of optimal solution and the second bullet guarantees the smoothness of the local objectives.
Assumption~\ref{Assump:ObjInfity} is a technical result coming out of our proofs and guarantees that, at infinity, the growth rate of the objective function is at least faster than linear. 
We note that Assumption~\ref{Assump:ObjInfity} is a mild assumption, which is evidenced by the following lemma (proof details are relegated to Appendix~\ref{appdx:proof_lemma2}).
\begin{lem}\label{Lemma:2}
Any strictly convex function $f(\cdot)$ satisfying Assumption~\ref{Assump:Obj} also satisfies Assumption~\ref{Assump:ObjInfity}.
\end{lem}

\begin{figure}[t!]
\centering
\begin{minipage}[t]{0.24\textwidth}
\centering
\includegraphics[width=4.5cm]{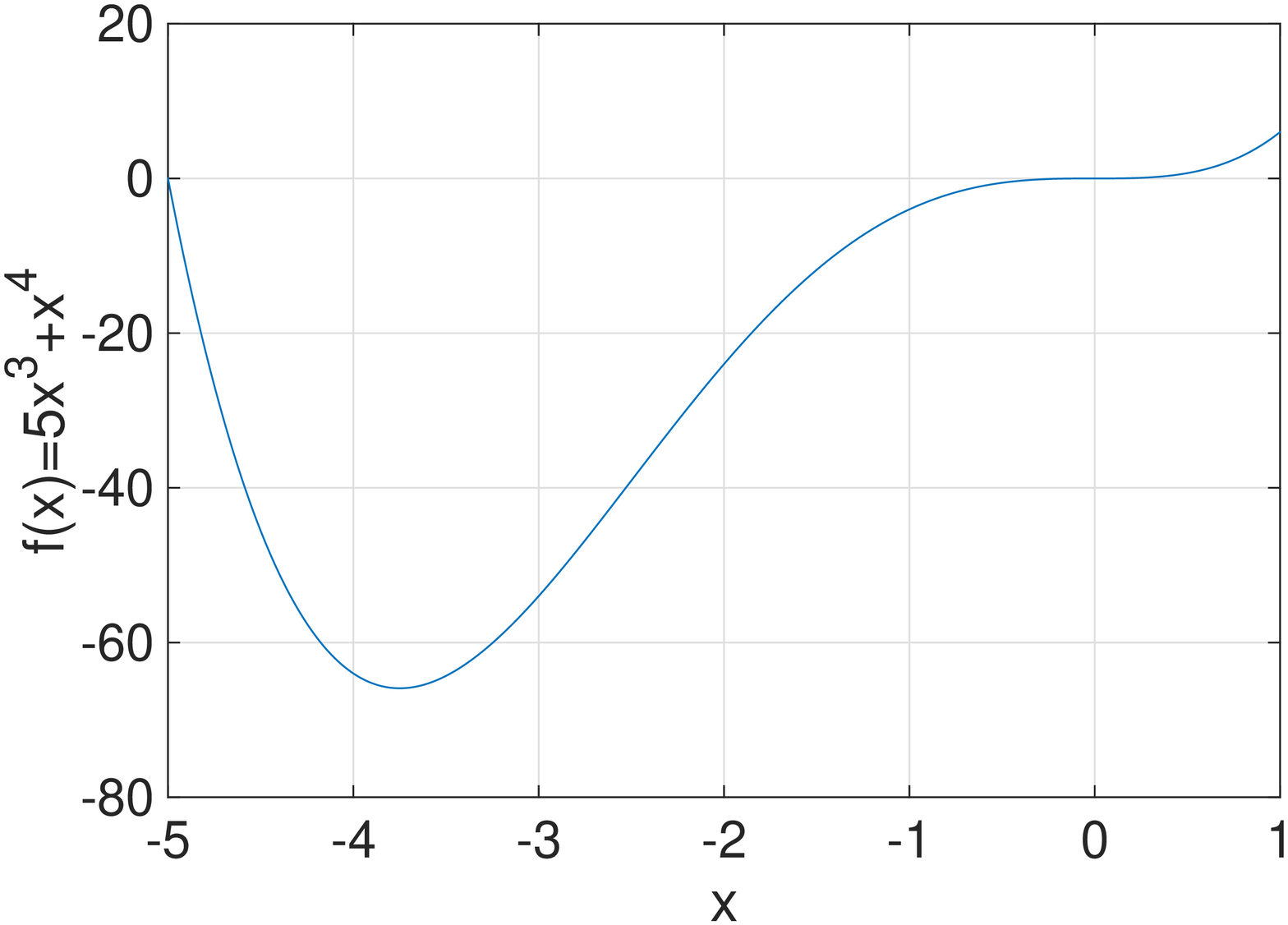}
\end{minipage}
\begin{minipage}[t]{0.24\textwidth}
\centering
\includegraphics[width=4.5cm]{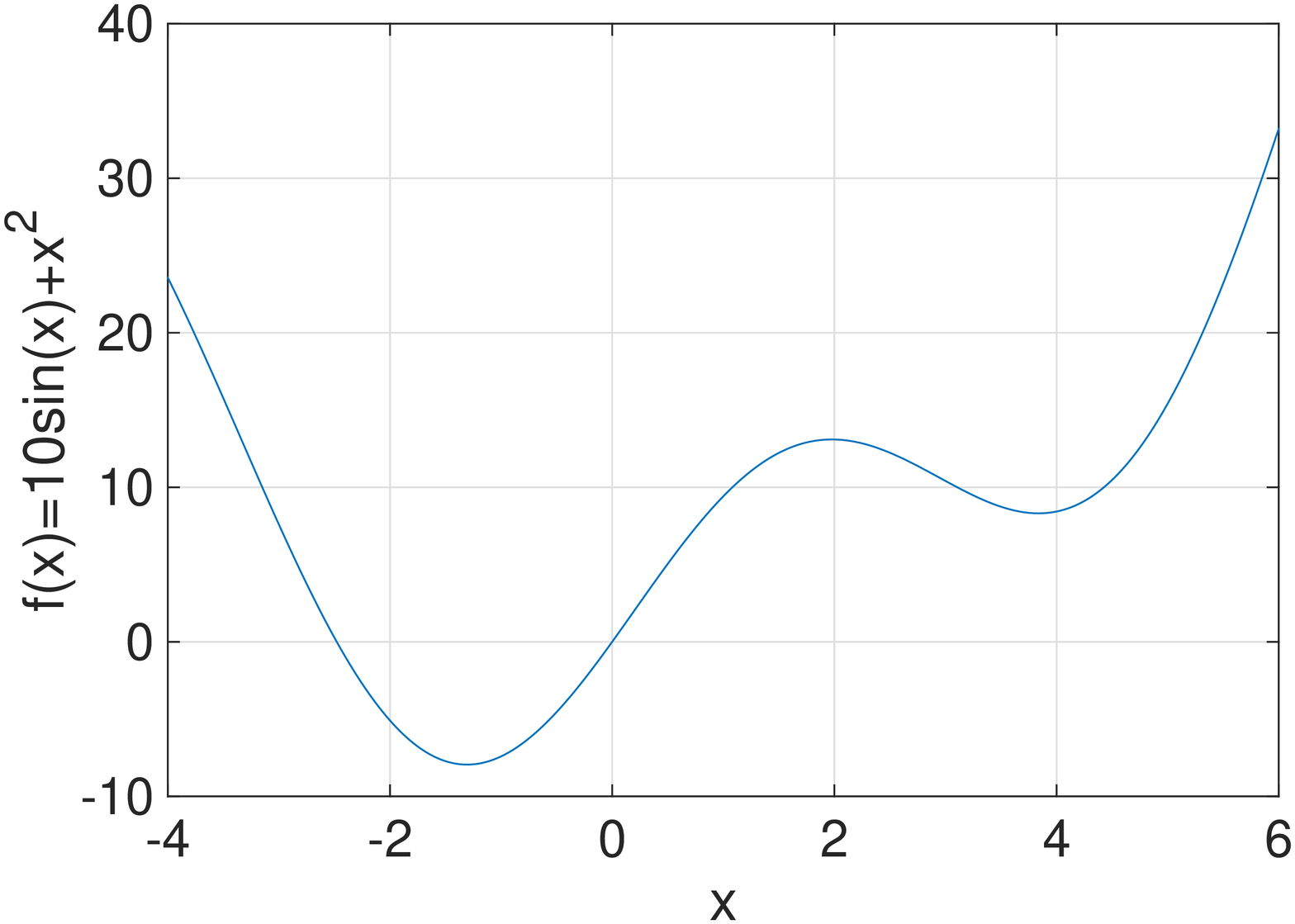}
\end{minipage}
\caption{Examples of non-convex functions satisfying Assumption~\ref{Assump:ObjInfity}.}
\label{fig:NonconvexObj}
\vspace{-.15in}
\end{figure}

In addition to convex objectives, many non-convex functions also satisfy Assumption~\ref{Assump:ObjInfity}, , as shown below and in Fig.~\ref{fig:NonconvexObj}:
\begin{ex}(Non-convex functions satisfying Assumption~\ref{Assump:ObjInfity}):
\begin{list}{\labelitemi}{\leftmargin=1em \itemindent=-0.5em \itemsep=.2em}
	\item $f(x)=x^4+5x^3,$ with $\lim_{x\rightarrow\infty}\|x\|/f(x) = 0$ but $\nabla^2 f(x)=12x^2+30x$ is smaller than $0$ when $x=-1;$
	\item $f(x)=10\sin(x)+x^2,$ with $\lim_{x\rightarrow\infty}\|x\|/f(x) = 0$ but $\nabla^2 f(x)=-10\cos(x)+2$ is smaller than $0$ when $x=0.$
\end{list}
\end{ex}

Our first key result is on the convergence of local variables to the mean vector across nodes:
\begin{thm}\label{thm:conv_meanvec}
Let the mean vector at the $k$-th iteration be defined as $\bar{\x}^k=\1\otimes \bar{x}^k \in\mathbb{R}^{NP},$ with $\bar{x}^k=\frac{1}{N}\sum_{i=1}^{N}x_{i,k}.$ 
Under Assumptions~\ref{Assump:Obj}, if $\mathbb{E}[\|\nabla \f(\x^i)\|]$ is bounded by $D$ and the amplifying exponent is $\gamma,$ then:
\begin{list}{\labelitemi}{\leftmargin=1em \itemindent=-0.5em \itemsep=.2em}
\item For constant step-size $\alpha_i=\alpha$, $\forall i$,  $\mathbb{E}[\|\x^k-\bar{\x}^k\|] \le \frac{\alpha D}{1-\beta}+O(\frac{\sqrt{NP}\sigma}{k^\gamma})$; 
\item For diminishing step-size $\alpha_i=O(\frac{1}{k^\eta})$ with some $\eta>0$, $\mathbb{E}[\|\x^k-\bar{\x}^k\|] = O(\frac{1}{k^{\min(\eta,\gamma)}})$. 
\end{list}
\end{thm}

\begin{rem}
{\em Theorem~\ref{thm:conv_meanvec} says that the local copies will converge to the mean vector asymptotically with a diminishing step-size, or stay within a bounded error ball of the mean vector if a constant step-size is adopted.}
\end{rem}

Our second key convergence result is on the convergence rate of ADC-DGD under constant step-sizes:
\begin{thm}[Constant Step-Size]\label{thm:const_step}
Let the step-size be constant, i.e.,$\alpha_{k} = \alpha$, $\forall k$, with $\alpha<\frac{1+\lambda_N(\mathbf{W})}{L}$.
Under Assumptions~\ref{Assump:Obj}-\ref{Assump:ObjInfity}, if the amplified exponent $\gamma >\frac{1}{2}$, then it holds that:
\begin{multline}\label{Eq:Theorem3}
\min_{k} \Big\{\mathbb{E}\Big[ \Big\|\frac{1}{N}\sum_{i=1}^{N} \nabla f_i(\bar{x}^{k}) \Big\|^2 \Big]\Big\} \le \\\frac{C_1\alpha^2}{N(1-\beta)^2} 
+ \frac{C_2}{\alpha k}+O\Big(\frac{2L\sqrt{P}\sigma}{\sqrt{N}k^\gamma} \Big),
\end{multline}
where $C_1 \triangleq L^4(B+\|\x^*\|)^2$ and $C_2 \triangleq 2[\frac{1}{N}\sum_{i=1}^{N} f_i(0) - \frac{1}{N}\sum_{i=1}^{N} f_i(x^*)]$ are two constants.
\end{thm}

\begin{rem}\label{Cor:1}
{\em Under the same conditions of Theorem \ref{thm:const_step}, we immediately have that Algorithm~2 has an ergodic convergence rate $O\big(\frac{1}{k^{\min{(1,\gamma)}}} \big)$ until $\min_{k} \{\mathbb{E}[\|\frac{1}{N}\sum_{i=1}^{N}\nabla f_i(\bar{x}^{k})\|^2\}$ reaching the error ball $O((\frac{\alpha}{1-\beta})^2)$ and the fastest rate is $O(\frac{1}{k}).$}
\end{rem}

Our third key convergence result is concerned with the convergence rate of ADC-DGD under diminishing step-sizes:

\begin{thm}[Diminishing Step-Sizes] \label{thm:diminish_step}
Under Assumptions \ref{Assump:Obj}-\ref{Assump:ObjInfity}, if the local objectives have bounded graidents, i.e. there exists a positive constant $D$ such that $\|\nabla f_i(x)\|\le D,$ $\forall x,$  and $\gamma>\frac{1}{2}$, $\eta\ge \frac{1}{2},$ then with diminishing step-size $\alpha_k=O(\frac{1}{k^\eta}),$ it holds that $\|\frac{1}{N}\sum_{i=1}^{N} \nabla f_i(\bar{x}^{k})\|^2=o(\frac{1}{k^{1-\eta}})$  almost surely.
\end{thm}

\begin{rem}
{\em In Theorem \ref{thm:diminish_step}, the exponent for the diminishing rate of step-size is lower bounded ($\eta\ge \frac{1}{2})$. Thus, the best convergence rate for this algorithm is $o(1/\sqrt{k}),$ which is faster than the rate $O(\log(k)/\sqrt{k})$ in \cite{tang2018decentralization}.
We also note that our convergence result is in ``Small-O'', which is stronger than conventional ``Big-O'' convergence results. 
}
\end{rem}

\begin{rem}[Intuition and Design Rationale of ADC-DGD]
{\em To understand why ADC-DGD converges, a closer look at (\ref{eqn:adc_dgd}) in Algorithm~2 reveals that:
\begin{align}\label{Eq:ActNoise}
&\tilde{x}_{j,k} \!=\! \tilde{x}_{j,k-1} \!+\! C(k^\gamma y_{j,k})/k^\gamma \!=\! \tilde{x}_{j,k-1}+(k^\gamma y_{j,k} \!+\! \epsilon_{k^\gamma y_{j,k}})/k^\gamma \nonumber\\
&=\tilde{x}_{j,k-1}+ y_{j,k}+(\epsilon_{k^\gamma y_{j,k}})/k^\gamma =x_{j,k}+(\epsilon_{k^\gamma y_{j,k}})/k^\gamma.
\end{align} 
Thanks to the properties of the unbiased stochastic operator (cf. Definition~\ref{defn:usco}), the noise term in the last step of (\ref{Eq:ActNoise}) has zero mean and a {\em vanishing} variance $\frac{\sigma^2}{k^{2\gamma}}$ as $k$ gets large.
This is in contrast to the accumulated non-vanishing noise term in DGD (cf. Eq.~(\ref{eqn:DGD_diverge})).
Eq.~(\ref{Eq:ActNoise}) also shows that our ADC-DGD algorithm can be interpreted as a {\em variance reduction} method.
Indeed, our proofs in Section~\ref{subsec:conv_analysis} are based on these intuitions.
}
\end{rem}

\subsection{Proofs of the Main Theorems} \label{subsec:conv_analysis}

Due to space limitation, in this subsection, we outline the key steps of the proofs of Theorems~\ref{thm:conv_meanvec}--\ref{thm:diminish_step}.
We relegate proof details to appendices.
Some appendices provide proof sketches due to the lengths of the proofs.

\smallskip
{\em Step 1): Introducing a Lyapunov Function.}
Consider the following Lyapunov function, which is also used in\cite{yuan2016convergence,berahas2017balancing}:
\begin{equation}\label{Eq:Lyap}
L_{\alpha_k}(\x)=\frac{1}{2}\x^{\top}(\I-\Z)\x+\alpha_k \1^{\top} \f(\x),
\end{equation}
where $\x=[x_1^{\top},\cdots,x_N^{\top}]^{\top},$ $\Z \triangleq \W \otimes \I_{P},$ and $\f(\x) \triangleq [f_1(x_1),\cdots,f_N(x_N)]^{\top}$ so that $\1^{\top} \f(\x)=\sum_{i=1}^{N}f_i(x_i)$.  
The following lemma is from\cite{yuan2016convergence}, which says that the Lyapunov function $L_a(\x)$ has Lipschitz-continuous gradient.

\begin{lem}\label{Lemma:1}
Under Assumption \ref{Assump:Obj}, the Lyapunov function
$L_{\alpha}(\x)=\frac{1}{2}\x^T(\I-\Z)\x+\alpha \1^T \f(\x)$ has $(1-\lambda_N(\mathbf{W})+\alpha L)$-Lipschitz gradient, i.e. $\|\nabla L_{\alpha}(\x) - \nabla L_{\alpha}(\mathbf{y})\|\le (1-\lambda_N(\mathbf{W})+\alpha L)\|\x-\mathbf{y}\|,$ $\forall \x,\mathbf{y} \in \mathbb{R}^{np}.$
\end{lem}

Note that, using the notation $\Z$, we can compactly rewrite the updating step (\ref{eqn:adc_dgd}) in Algorithm~2 as follows:
\begin{align}\label{Eq:updating}
\x^{k+1}&=\Z(\tilde{\x}^k+\mathbf{d}^k/k^\gamma)-\alpha_k \nabla \f(\x^k)\notag\\ 
&=\Z\x^k+\Z\epsilon^k/k^\gamma-\alpha_k \nabla \f(\x^k)\notag\\
&=\x^k-[(\I-\Z)\x^k+\alpha_k \nabla \f(\x^k)]+\Z\epsilon^k/k^\gamma \notag\\
&=\x^k- \nabla L_{\alpha_k}(\x^k)+\Z\epsilon^k/k^\gamma,
\end{align}
where 
$\nabla \f(\x^k)=[\nabla f_1(x_{1,k})^{\top},\cdots,\nabla f_N(x_{N,k})^{\top}]^{\top},$ $\x^k=(x_{1,k}^{\top},\cdots,x_{N,k}^{\top})^{\top}$ is the parameter in the $k$-th iteration,
$\tilde{\x}^k=[\tilde{x}_{1,k}^{\top},\cdots,\tilde{x}_{N,k}^{\top}]^{\top}$ is the vector of imprecise parameters, and 
$\mathbf{d}^k \triangleq [d_{1,k}^{\top},\cdots,d_{N,k}^{\top}]^{\top}$ and $\epsilon^k \triangleq [\epsilon_{k^\gamma y_{1,k}}^{\top},\cdots,\epsilon_{k^\gamma y_{N,k}}^{\top}]^{\top}$. 
It can be seen that Eq.~(\ref{Eq:updating}) is one-step stochastic gradient descent for $L_{\alpha_k}(\x)$ and the noise term $\Z\epsilon^k/k^\gamma$ has zero mean and variance with diminishing bound $NP\sigma^2/k^{2\gamma}$, i.e.,
\begin{align}
&\mathbb{E}[\Z\epsilon^k/k^\gamma]=\Z\mathbb{E}[\epsilon^k]/k^\gamma=0,\\
&\mathbb{E}[\|\Z\epsilon^k/k^\gamma\|^2]\le\|\Z\|^2\mathbb{E}[\|\epsilon^k\|^2]/k^{2\gamma} \stackrel{(a)}{\le} NP\sigma^2/k^{2\gamma},\label{Eq:Var}
\end{align}
where $(a)$ follows from the fact that the eigenvalues of $\Z$ are in $(-1,1]$ and $\epsilon^k\in \mathbb{R}^{NP}.$

{\em Step 2) Convergence of the Objective Value.}
Note from (\ref{Eq:ActNoise}) that the noise caused by compression is similar to the noise in the standard stochastic gradient descent method (SGD). 
Hence, we can apply similar analysis techniques from SGD on the iterations of ADC-DGD to obtain the following results:

\begin{thm}[Bounded Gradient]\label{thm:bnd_grad}
Under Assumptions \ref{Assump:Obj}-\ref{Assump:ObjInfity}, if the step-size $\alpha<\frac{1+\lambda_N(\mathbf{W})}{L}$ and the amplified exponent $\gamma > \frac{1}{2}$ in Algorithm~2, 
then there exists a constant $B>0$ such that $\mathbb{E}[\|\x^k\|]\le B$ and $\mathbb{E}[\|\nabla \f(\x^k)\|]\le L(B+\|\x^*\|),$ where $\x^{*}=\1\otimes x^*\in \mathbb{R}^{NP}.$ Moreover, $\mathbb{E}[\| \nabla L_{\alpha}(\x^k)\|^2] = o(1/k).$
\end{thm}

Theorem \ref{thm:bnd_grad} shows that with an appropriate step-size and an amplifying exponent, Algorithm~2 converges. 
But due to the compression noise, the convergence rate is sublinear.
To see this, note that $\nabla L_{\alpha}(\x)=(\I-\Z)\x+\alpha \nabla \f(\x),$ and $\1^{\top}(\I-\Z)=0.$ 
Thus, $\1^{\top}\nabla L_{\alpha}(\x)=\alpha \sum_{i=1}^{N}\nabla f_i(x_i),$ which implies $\|\alpha \sum_{i=1}^{N}\nabla f_i(x_i)\|^2\le \|\nabla L_{\alpha}(\x)\|^2.$ From Theorem \ref{thm:bnd_grad}, the convergence rate of $\mathbb{E}[\|\sum_{i=1}^{N}\nabla f_i(x_i)\|^2]$ is also $o(1/k).$

\smallskip
{\em Step 3) Proving Theorem~\ref{thm:conv_meanvec}}.
Note from Algorithm~2 and (\ref{Eq:updating}) that the following hold:
\vspace{-.05in}
\begin{align} \label{Eq:recursive}
\!\!\!\!\!\!\begin{cases}
&\x^1= \Z\x^0-\alpha_0 \nabla \f(\x^0) = -\alpha_0 \nabla \f(\x^0), \\
&\x^2= \Z\x^1-\alpha_1 \nabla \f(\x^1) + \Z\frac{\epsilon^1}{1^\gamma}, \\
&~~~= -\alpha_0 \Z\nabla \f(\x^0) - \alpha_1 \nabla \f(\x^0) + \Z\epsilon^1, \\ 
&\hspace{.21in}\vdots\\
&\x^k  =-\sum_{i=0}^{k-1}\alpha_i \Z^{k-i-1}\nabla \f(\x^i) + \sum_{i=1}^{k-1} \Z^{k-i}\frac{\epsilon^i}{i^\gamma}.
\end{cases}
\end{align}
Eq.~(\ref{Eq:recursive}) characterize the trajectory of the iterates. Each iterate consists of two parts, one from gradients and the other from noises. 
Note that in (\ref{Eq:recursive}), the variance of accumulated noises are in the form of $h_k \triangleq \sum_{i=1}^{k}\frac{\beta^{k-i}}{i^\gamma}$. 
Next, we prove an interesting lemma for $h_k$, which is useful in proving Theorem~\ref{thm:conv_meanvec}.
\begin{lem}\label{Lemma:convolution}
Define $h_k \triangleq \sum_{i=1}^{k}\frac{\beta^{k-i}}{i^\gamma},$ where $\beta \in [0,1)$ and $\gamma>0.$ It follows that $h_k= O(\frac{1}{k^\gamma}).$
\end{lem}
Lemma \ref{Lemma:convolution} implies that the negative effect of  compression noises can be ignored asymptotically, which induces convergence.
With (\ref{Eq:recursive}), Theorem \ref{thm:bnd_grad} and Lemma~\ref{Lemma:convolution}, 
we can finally prove Theorem~\ref{thm:conv_meanvec} and the details are relegated to Appendix~\ref{appdx:conv_meanvec}.

\smallskip
{\em Step 4) Proving Theorems~\ref{thm:const_step} and \ref{thm:diminish_step}}. With some algebraic derivation, we can show the following fundamental result:

\begin{lem}\label{Lemma:5}
Let $\mathcal{F}_k \!=\! \sigma\langle \x^1,\!\cdots\!, \x^k \rangle$ be a filtration.
Under Assumptions \ref{Assump:Obj}-\ref{Assump:ObjInfity}, the following inequality holds:
{\small{\begin{align}\label{Eq:lemma5}
&\mathbb{E}\Big[\frac{1}{N}\sum_{i=1}^{N} f_i(\bar{x}^{k+1}) \Big|\mathcal{F}_k \Big]+ \frac{\alpha_k}{2} \Big\|\frac{1}{N}\sum_{i=1}^{N} \nabla f_i(\bar{x}^{k}) \Big\|^2 \notag\\
&\le \frac{1}{N}\sum_{i=1}^{N} f_i(\bar{x}^{k})  + \frac{\alpha_k L^2}{2N^2}\sum_{i=1}^{N} \Big\| \bar{x}^{k}- x_{i,k} \Big\|^2 + \frac{LP\sigma^2}{2Nk^{2\gamma}},
\end{align}}}
where $\alpha_k$ is the step-size at the $k$-the iteration. 
\end{lem}

Eq.~(\ref{Eq:lemma5}) in Lemma \ref{Lemma:5} is similar to the contraction in stochastic gradient descent algorithm, which relates the objective values and gradient norm. 
Then, by telescoping and the supermartingale convergence theorem, we can prove Theorems~\ref{thm:const_step} and \ref{thm:diminish_step} (see Appendices~\ref{appdx:thm_const_step} and \ref{appdx:thm_diminish_step}).

\subsection{Understanding the Role of the Amplifying Exponent $\gamma$}\label{sec:gamma}
In our algorithm, the amplifying exponent $\gamma$ is a key component to adjust the communication rate. 
From Theorems~\ref{thm:const_step} and \ref{thm:diminish_step}, it can be seen that within $(1/2,1],$ the larger $\gamma$ means the faster convergence. 
However, since the transmitted value is $C(k^\gamma {\mathbf{y}}^k),$ we can see that a larger $\gamma$ leads to a larger $k^\gamma {\mathbf{y}}^k,$ which may lead to overflow error (for example, type `int8' in Matlab could only present data within $(-128,127)$). 
Hence, it is necessary to guarantee that $k^\gamma {\mathbf{y}}^k$ would not grow too fast. 
Recalling Eqs.~(\ref{Eq:updating}) and (\ref{eqn:adc_dgd}) in Algorithm~2, we have 
\begin{align*}
\y^k&={\x}^k-\tilde{\x}^{k-1}=(\I+\Z){\epsilon}^k/k^\gamma-\nabla L_{\alpha_k}({\x}^k).
\end{align*}
Under the expectation, the transmitted value is bounded by
\begin{align*}
\mathbb{E}[\|k^\gamma \y^k\|]&=\mathbb{E}[\|(\I+\Z){\epsilon}^k-k^\gamma\nabla L_{\alpha_k}({\x}^k)\|]\\
&\le \mathbb{E}[\|(\I+\Z){\epsilon}^k\|+\|k^\gamma\nabla L_{\alpha_k}({\x}^k)\|]\\
&\le 2\mathbb{E}[\|{\epsilon}^k\|]+k^\gamma\mathbb{E}[\|\nabla L_{\alpha_k}({\x}^k)\|].
\end{align*}
From Definition~\ref{defn:usco}, we have that each element of $\mathbb{E}[\|{\epsilon}^k\|]$ is bounded by $\sigma.$ 
From Theorem \ref{thm:bnd_grad}, we have $\mathbb{E}[\|\nabla L_{\alpha_k}({\x}^k)\|^2]=o(1/k).$ 
Thus, $\mathbb{E}[\|k^\gamma {\mathbf{y}}^k\|]$ is bounded by $o(k^{\gamma-\frac{1}{2}})$.
We state this result in the following proposition:
\begin{prop}\label{prop:trans_signal}
Under Assumptions \ref{Assump:Obj}-\ref{Assump:ObjInfity}, with $\gamma>\frac{1}{2},$ the transmitted value $k^\gamma {\mathbf{y}}^k$ satisfies $\mathbb{E}[\|k^\gamma {\mathbf{y}}^k\|]=o(k^{\gamma-\frac{1}{2}}).$
\end{prop}
The insight from Proposition~\ref{prop:trans_signal} is that with $\gamma \in (\frac{1}{2},1],$ the growth speed for the transmitted value is slower than $o(\sqrt{k}),$ which is not very fast.

\section{Numerical Results} \label{sec:numerical}

In this section, we will present several numerical experiments to further validate the performance of ADC-DGD. 

\smallskip
{\em 1) Effect of Compression:} 
First, we compare ADC-DGD with some existing methods to show its convergence rate and communication-efficiency.
Consider a four-node network as shown in Fig.~\ref{Fig:network} with the following  global objective function: $\min_x f(x)=f_1(x)+f_2(x)+f_3(x)+f_4(x)$, 
where $f_1(x)=-4x^2$, $f_2(x)= 2(x-0.2)^2$, $f_3(x)= 2(x+0.3)^2$, and $f_4(x)= 5(x-0.1)^2$.
It can be seen that $f_1(x)$ is non-convex, while the rest are convex.
The communication consensus matrix used in this experiment is shown in Fig.~\ref{Fig:mixing_matrix}.
%

\begin{figure}[t!]
\centering
\begin{minipage}[t]{0.21\textwidth}
\centering
\includegraphics[width=0.54\textwidth]{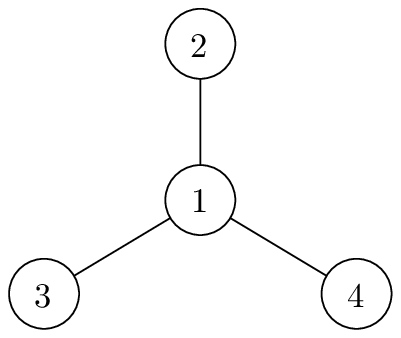}
\caption{A four-node network.}
\label{Fig:network}
\end{minipage}
\hspace{.1in}
\begin{minipage}[t]{0.25\textwidth}
\vspace{-.68in}
\begin{align}
{\small 
\mathbf{W}=\left[\begin{matrix}
  1/4 &1/4&1/4&1/4\\
   1/4&3/4&0&0\\
   1/4&0&3/4&0\\
   1/4&0&0&3/4
  \end{matrix}\right], \notag
}
  \end{align}
 \caption{The consensus matrix for Fig.~\ref{Fig:network}.}
 \label{Fig:mixing_matrix}
\end{minipage}
\vspace{-.1in}
\end{figure}

In our simulation, we compare our ADC-DGD with the conventional DGD and $\text{DGD}^t.$
For $\text{DGD}^t,$ we consider two cases: $t=3$ and $t=5.$
In ADC-DGD, the amplifying exponent $\gamma$ is set to $1$. 
We use two step-size strategies: 1) constant step-size $\alpha$ (i.e. $\eta=0$) and 2) diminishing step-size $\alpha/\sqrt{k}$ (i.e. $\eta=1/2$).
We adopt the quantized operator in \cite{alistarh2017qsgd} as the compression operator. 
After compression, the values are integer. 
Hence, they can be stored as type `int16', which is 2 bytes. 
However, the uncompressed values are stored as type `double', costing 8 bytes. 
The convergence results for one trial are illustrated in Fig.~\ref{Fig:IL} and Fig.~\ref{Fig:BL}. 
 
\begin{figure}[t!]
\centering
\begin{minipage}[t]{0.24\textwidth}
\centering
\includegraphics[width=1\textwidth]{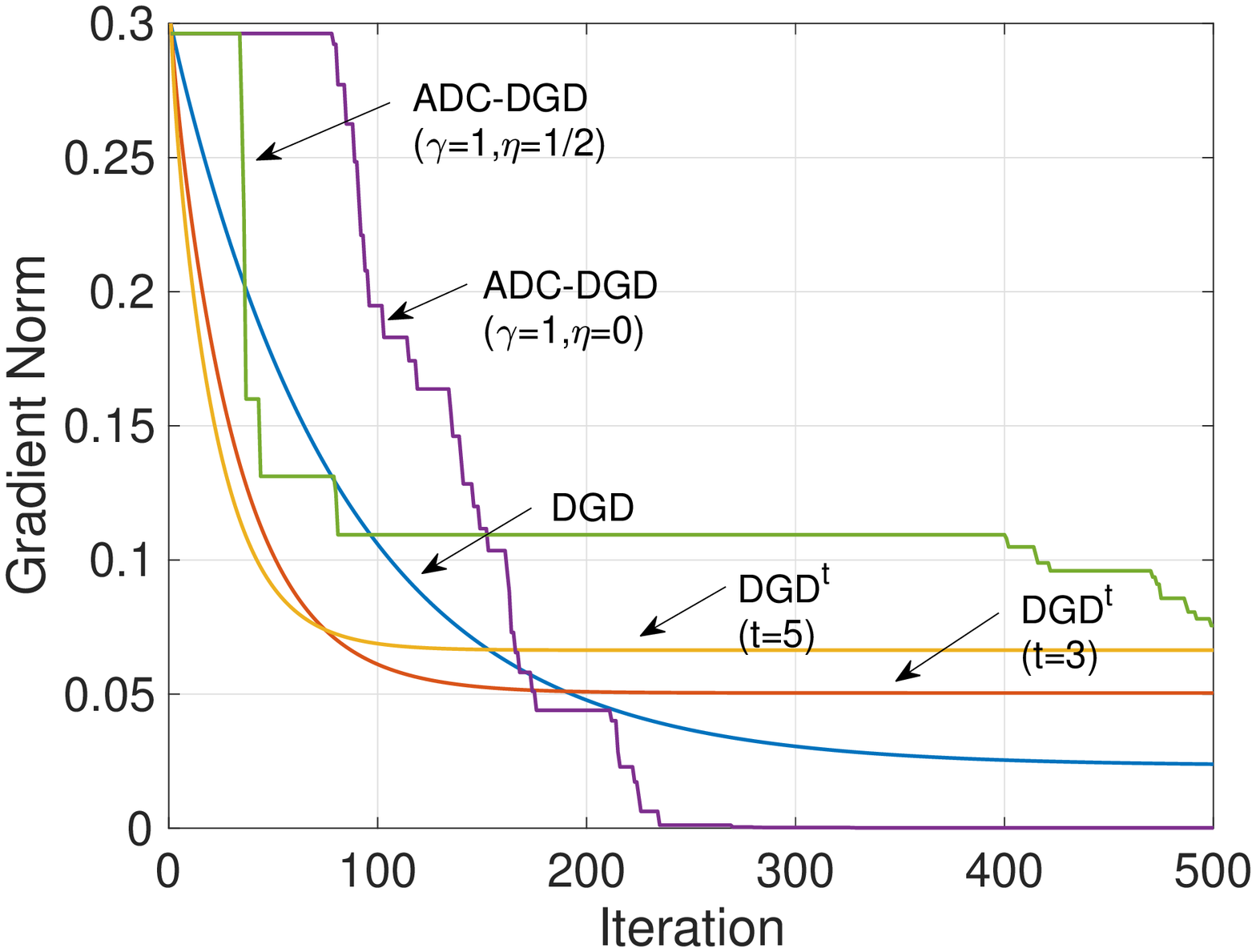}
\caption{Convergence comparisons between ADC-DGD, DGD, and DGD$^t$.}
\label{Fig:IL}
\end{minipage}
\begin{minipage}[t]{0.24\textwidth}
\centering
\includegraphics[width=1\textwidth]{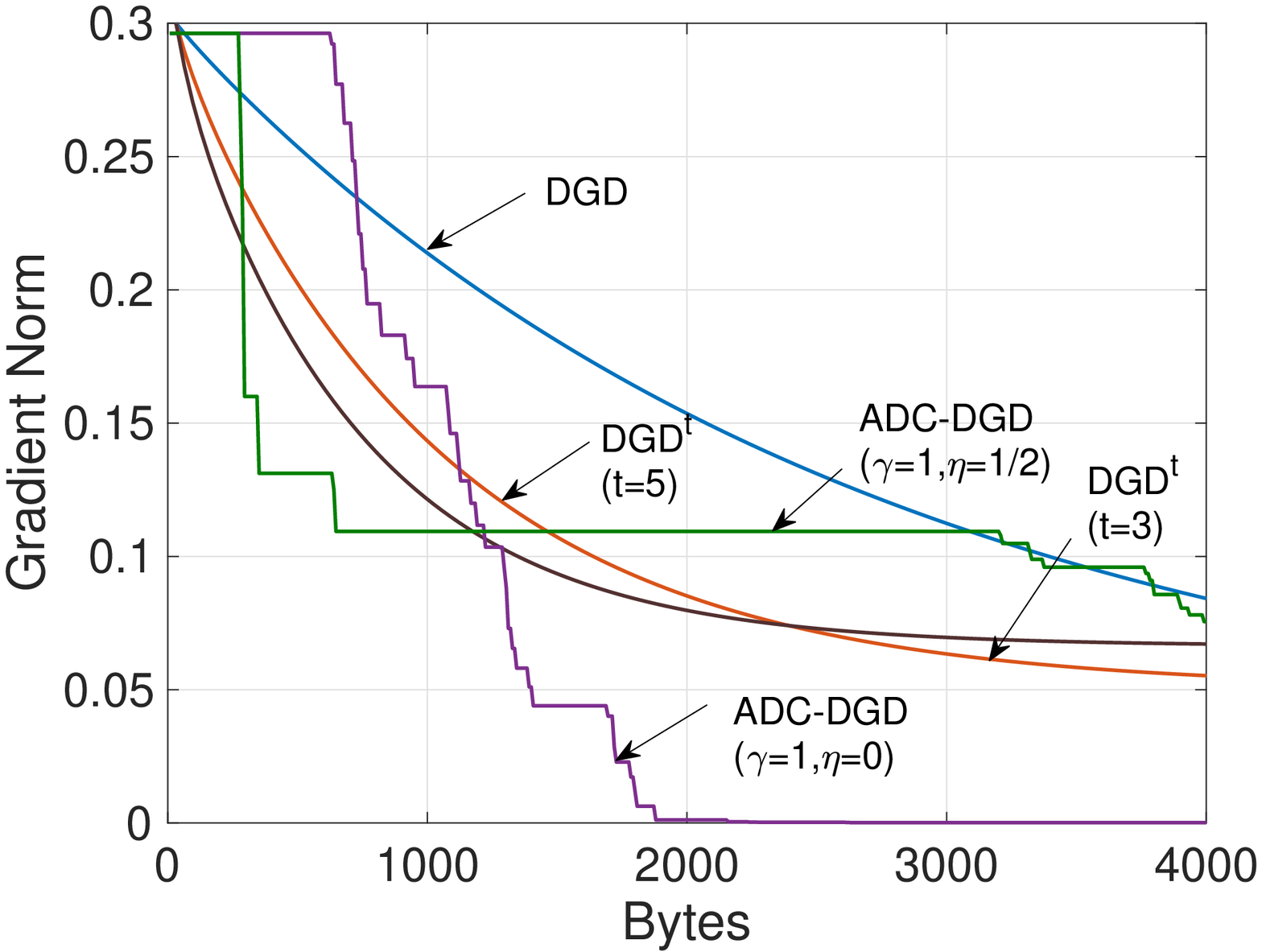}
\caption{Amount of exchanged information (bytes) vs gradient norm.}
\label{Fig:BL}
\end{minipage}
\vspace{-.2in}
\end{figure}

From the simulations, we can see that: 
1) with a fixed step-size, all algorithms converge to an error ball, while the radiuses of the conventional DGD and ADC-DGD are relatively smaller. 
This is because with a larger $t,$ $\beta^t$ becomes smaller and hence the error ball $O(\alpha/(1-\beta^t))$ for $\text{DGD}^t$ becomes larger; 
2) By using compression, the convergence process of ADC-DGD is relatively less smooth. 
But the compression noise does not affect convergence. 
With the same step-size, the conventional DGD and ADC-DGD have the almost the {\em same} convergence rate;
3) By using diminishing step-sizes, the convergence speed for ADC-DGD becomes slower. 
However, the objective value remains decreasing;
4) By comparing the amount of exchanged information, ADC-DGD with the fixed step-size converges the fastest, using only 2000 bytes. 
This shows that our algorithm is the {\em most} communication-efficient.
 
\smallskip
{\em 2) Effect of the Amplifying Exponent:}
Next, we show the effect of the amplifying exponent $\gamma.$ 
As discussed in Section \ref{sec:gamma}, with a small $\gamma,$ the noise caused by compression could lead to a slow convergence.
On the other hand, with a large $\gamma,$ the transmitted value $k^\gamma y$ could be too large and cause overflow, especially for quantized compressed operator. 
Here, we change $\gamma$ using $\{0.6,0.8,1.0,1.2\}$ and keep the rest of the parameters the same.
For each $\gamma,$ we repeat the algorithm $100$ times and compute the average objective values, as well as the maximum transmitted value from all the nodes in each iteration. 
The simulation results are shown in Figs.~\ref{Fig:GE} and \ref{Fig:GY}.
We can see that, with a larger $\gamma$ value, the algorithm converges faster and the curve is smoother, while the transmitted values are increasing a little bit faster. 
In this example, we can see that $\gamma=0.8$ strikes a good balance between convergence and maximum transmitted value.

\begin{figure}[t!]
\centering
\begin{minipage}[t]{0.24\textwidth}
\centering
\includegraphics[width=1\textwidth]{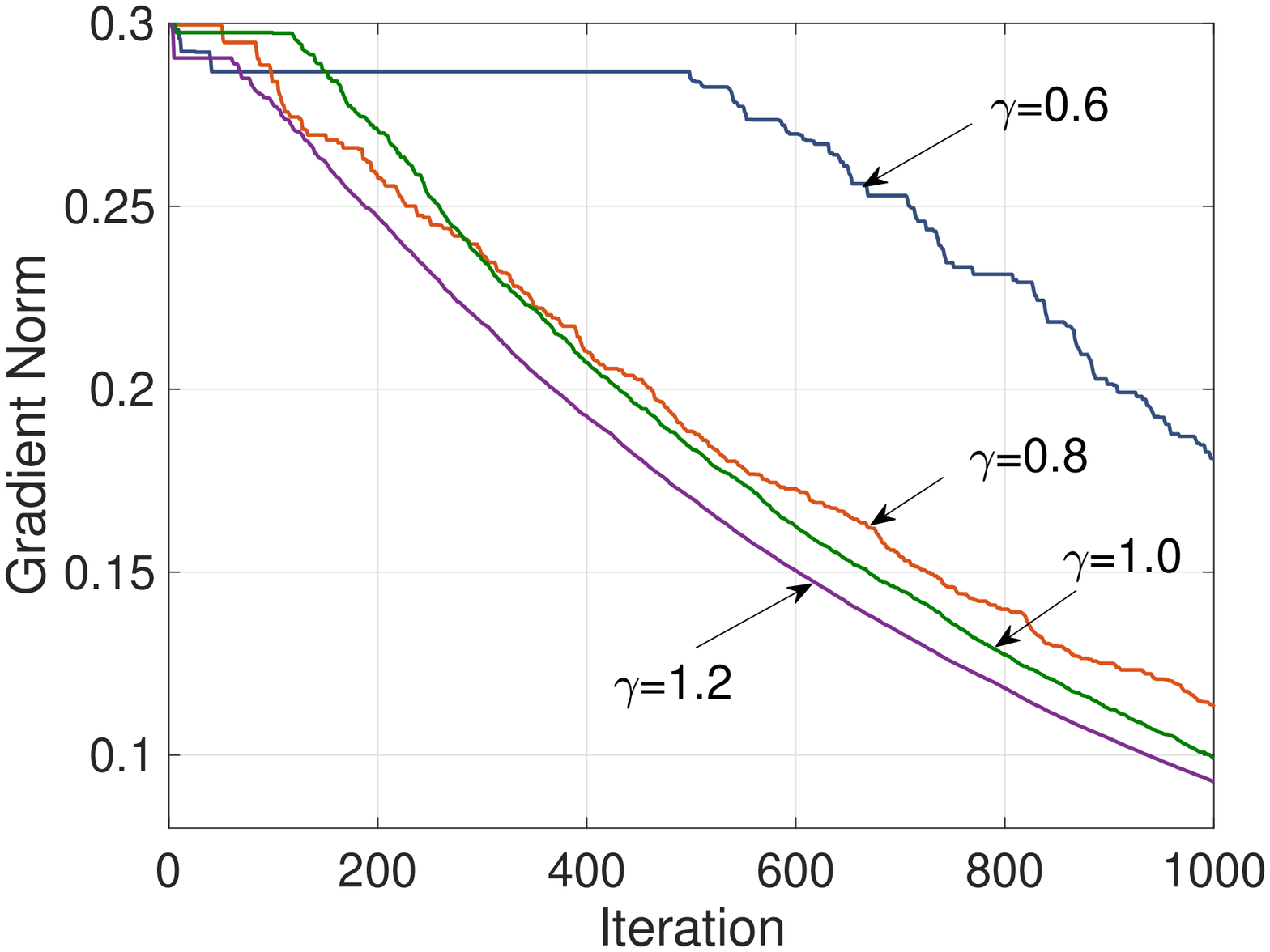}
\caption{Convergence behaviors under different choices of $\gamma$.}
\label{Fig:GE}
\end{minipage}
\begin{minipage}[t]{0.24\textwidth}
\centering
\includegraphics[width=1\textwidth]{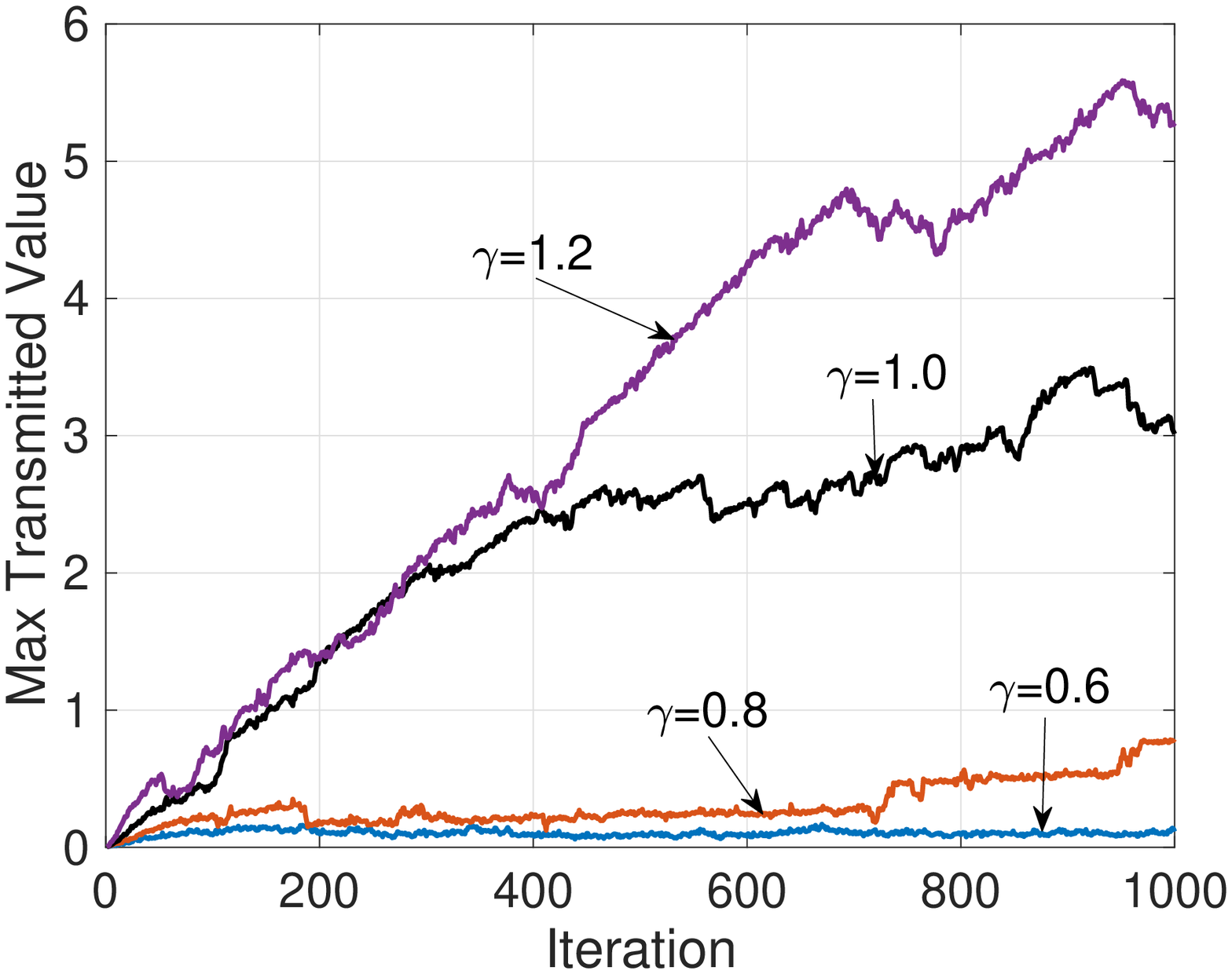}
\caption{Growth of transmitted values vs. number of iterations.}
\label{Fig:GY}
\end{minipage}
\vspace{-.1in}
\end{figure}

\smallskip
{\em 3) Effect of Network Size:}
The following  simulations indicate that our algorithm could be scaled to large-size networks. 
In our simulation, we consider the `circle' system: each node only connects with two neighboring nodes and forms a circle. 
For example, Fig.~\ref{Fig:cir} shows a five-node circle. 
We set $n$ to be $3$, $5$, $10$, $20$ in our experiment.
The local objectives are in the form of $f_i(x)=a_i(x-b_i)^2.$ 
In our simulation, $\{a_i,b_i\}_{i=1}^{n}$ are independently randomly generated: $a_i\sim \text{Uniform}[0,10]$ and $b_i \sim \text{Uniform}[0,1]$. 
For each value of $n$, we repeat $100$ trials and compute the average gradient norm.
The convergence results are shown in Fig.~\ref{Fig:node}. 
It can be seen that our algorithm works well as the network size increases, demonstrating the scalability of ADC-DGD. 

\begin{figure}[t!]
\centering
\begin{minipage}[t]{0.24\textwidth}
\centering
\vspace{-1.07in}
\includegraphics[width=2.6cm]{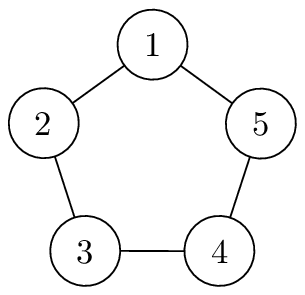}
\vspace{.08in}
\caption{The 5-node circle topology.}
\label{Fig:cir}
\end{minipage}
\begin{minipage}[t]{0.24\textwidth}
\centering
\includegraphics[width=1\textwidth]{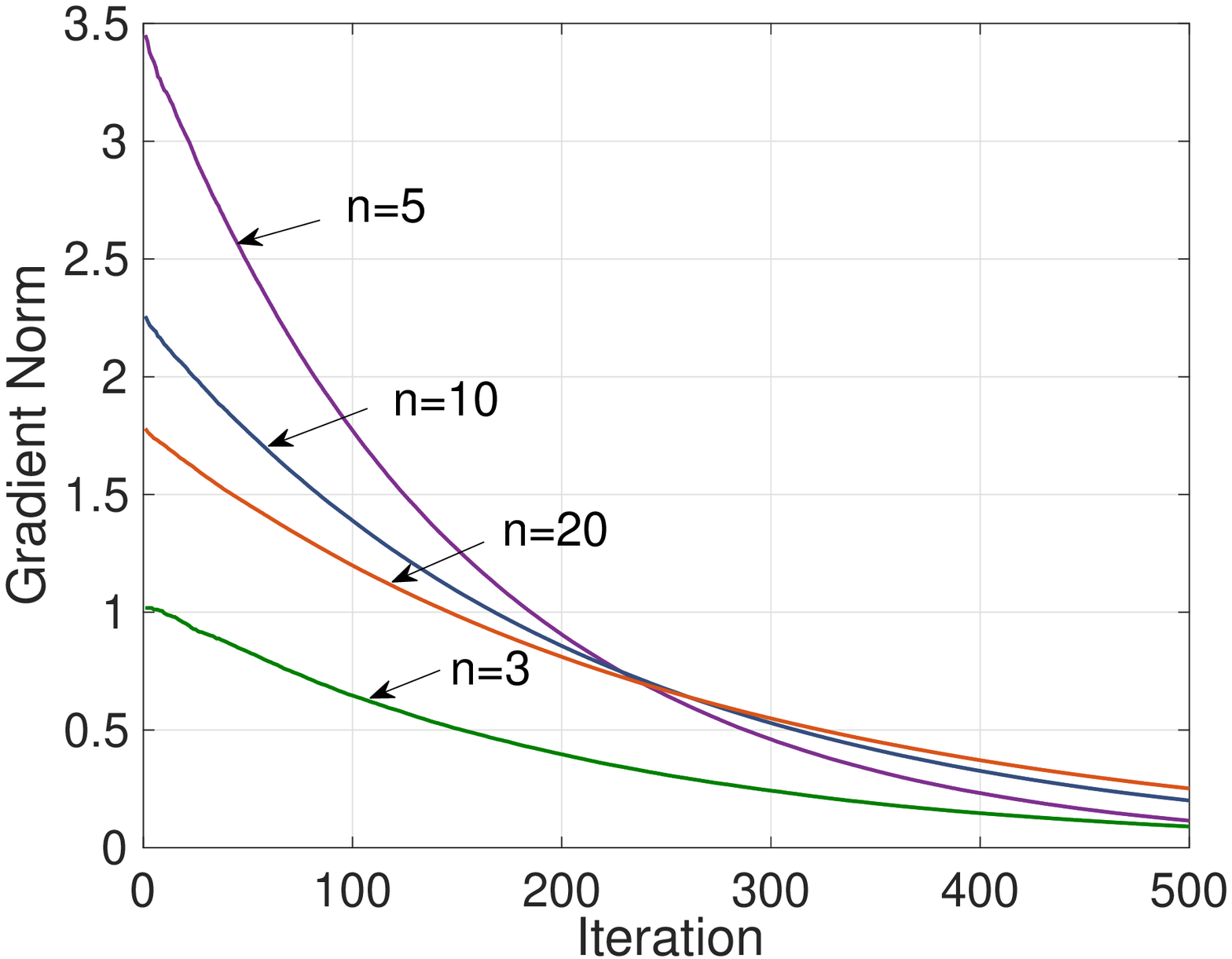}
\caption{The effect of network size.}
\label{Fig:node}
\end{minipage}
\vspace{-.2in}
\end{figure}

\section{Conclusion} \label{sec:conclusion}

In this paper, we considered designing communication-efficient network consensus optimization algorithms in networks with slow communication rates.
We proposed a new algorithm called amplified-differential compression decentralized gradient descent (ADC-DGD), which is based on compression to reduce communication costs. 
We investigated the convergence behavior of ADC-DGD on smooth but possibly non-convex objectives in this work.
We showed that: 1) by employing a fixed step-size $\alpha$, ADC-DGD converges with the $O(1/k^{\min(1, \gamma)})$ ergodic rate until reaching an error ball of size $O(\alpha^2)$ with the amplified parameter $\gamma$; 
2) ADC-DGD enjoys the best convergence rate $o(1/\sqrt{k})$ and converge to a stationary point almost surely with diminishing step-sizes. 
Consensus optimization with compressed information is an important and under-explored area.
An interesting future topic is to generalize our ADC-DGD algorithmic framework to analyze cases with local stochastic gradients, which could further lower the implementation complexity of ADC-DGD.

\bibliographystyle{IEEEtran}
\bibliography{reference.bib}

\vspace{12pt}

\appendices
\section{Proof for Lemma \ref{Lemma:2}} \label{appdx:proof_lemma2}
Without loss of generality, we prove the case of one-dimensional objective. Firstly, we consider $f_i(\cdot)$ reach the minimal at $x_i^*=0$ and $f_i(x_i^*)=0.$ With the convexity of $f_i(\cdot),$ $f_i(x)\le (1-\frac{1}{a})f_i(0)+\frac{1}{a}f_i(a x),$ $\forall x,~a>1.$ Consider $x_0$ with $f_i(x_0)\neq 0,$ and define $\max_i\{\frac{\|x_0\|}{f_i(x_0)}\}=M_0.$ Consider $\forall y>x_0>0,$ there exists a constant $a_y>1$ such that $y=a_yx_0$ and thus $f_i(x_0)\le \frac{1}{a_y}f_i(y).$ Hence, $\frac{\|y\|}{f_i(y)}\le  \frac{\|a_yx_0\|}{a_yf_i(x_0)}\le M_0.$ It is easy to obtain the same result for negative values $y.$ Therefore, $\frac{\sum_{i=1}^{N}\|x_i\|}{\sum_{i=1}^{N}f_i(x_i)}\le \sum_{i=1}^{N}\frac{\|x_i\|}{f_i(x_i)} \le NM_0:=M.$
Next, if $x_i^*\neq 0$ and $f_i(x_i^*)\neq 0.$ Consider the transformation, $g_i(x)=f_i(x+x_i^*)-f_i(x_i^*),$ then $x=0$ is the minimal solution for $g_i(x),$ $g_i(0)=0$ and also $g_i(x)$ maintains the convexity. From the above, we know that $\frac{\|x\|}{g_i(x)}\le M_0.$ Therefore, $\frac{\|x\|}{f_i(x+x_i^*)-f_i(x_i^*)}\le M_0.$ Denote $y=x+x_i^*,$ we have $\frac{\|y-x_i^*\|}{f_i(y)-f_i(x_i^*)}\le M_0.$ Consider the following limits:
\begin{align*}
&\lim_{\|y\|\rightarrow \infty} \!\! \frac{\|y\|}{f_i(y)} \!/\! \frac{\|y-x_i^*\|}{f(y)-f_i(x_i^*)} \!\!=\!\! \lim_{\|y\|\rightarrow \infty} \!\!\frac{\|y\|}{\|y-x_i^*\|} \!/\! \frac{f_i(y)}{f_i(y)-f_i(x_i^*)} \\
& = \lim_{\|y\|\rightarrow \infty}\frac{f_i(y)-f_i(x_i^*)}{f_i(y)} \le 1 + \lim_{\|y\|\rightarrow \infty} |\frac{f_i(x_i^*)}{f_i(y)}| \le 2,
\end{align*}
which implies that $\lim_{\|y\|\rightarrow \infty}\frac{\|y\|}{f_i(y)} \le 2M_0.$ 
Similarly, we can show $\frac{\sum_{i=1}^{N}\|x_i\|}{\sum_{i=1}^{N}f_i(x_i)}\le 2NM_0:=M.$
This completes the proof.

%

\section{Proof for Theorem \ref{thm:bnd_grad}} \label{appdx:bnd_grad}
First, we give the following useful lemma for series with converged infinite sum.

\begin{lem}\label{lemma:series}
Assume the infinite sum of the positive series $\{a_k\}_{k=1}^{\infty},$ $\sum_{k=1}^{\infty} a_k <\infty,$ then $a_k=o(1/k).$ 
\end{lem}
\begin{proof}
This result has been discussed in \cite{zhang2018taming}. We give the details in the following. 
Define $\{b_k\}_{k=1}^{\infty}$ with $b_k=1/k.$ With the properity of $p$-series, it is well-known that $\sum_{k=1}^{\infty}b_k=\infty.$ 
 
Now, if we assume $a_k$ is higher order than $o(1/k),$ then $\lim_{k\rightarrow \infty}a_k/b_k=c,$ where $c=(0,+\infty)\cup\{+\infty\}.$ With the limit comparison test, then $\sum_{k=1}^{\infty} a_k <\infty,$ which leads to a contradiction with the fact $\sum_{k=1}^{\infty} a_k <\infty.$ Hence, it requires that $c=0$ and $a_k=o(1/k).$
\end{proof}

\begin{proof}
Define the filtrations $\{\mathcal{F}_k = \sigma\langle \textbf{x}^1,\cdots, \textbf{x}^k \rangle\}$, $\mathcal{F}_k \subset \mathcal{F}_{k+1}.$ From Lemma \ref{Lemma:1}, given $\alpha,$ the Lyapunov function $L_{\alpha}(\textbf{x})$ has $(1-\lambda_n(\textbf{W})+\alpha L)$-Lipschitz gradient.
Hence, the following holds:{\small{
\begin{align*}
L_\alpha(\textbf{x}^{k+1}) &\stackrel{(a)}{\le} 
L_\alpha(\textbf{x}^{k}) + \langle \nabla L_\alpha(\textbf{x}^{k}),\textbf{x}^{k+1}-\textbf{x}^k \rangle \\
&+ \frac{1}{2}(1-\lambda_n(\textbf{W})+\alpha L)\|\textbf{x}^{k+1}-\textbf{x}^k\|^2\\
& \stackrel{(b)}{=} L_\alpha(\textbf{x}^{k}) - \langle \nabla L_\alpha(\textbf{x}^{k}), \nabla L_{\alpha}(\textbf{x}^k)-\frac{\textbf{W}\epsilon^k}{k^\gamma} \rangle \\
&+ \frac{(1-\lambda_n(\textbf{W})+\alpha L)}{2}\| \nabla L_{\alpha}(\textbf{x}^k)-\frac{\textbf{W}\epsilon^k}{k^\gamma}\|^2\\
& = L_\alpha(\textbf{x}^{k}) - \|\nabla L_\alpha(\textbf{x}^{k})\|^2 + \langle \nabla L_\alpha(\textbf{x}^{k}), \frac{\textbf{W}\epsilon^k}{k^\gamma} \rangle \\
&+ \frac{(1-\lambda_n(\textbf{W})+\alpha L)}{2}[\| \nabla L_{\alpha}(\textbf{x}^k)\|^2\\
&+\|\frac{\textbf{W}\epsilon^k}{k^\gamma}\|^2-2\langle  \nabla L_{\alpha}(\textbf{x}^k), \frac{\textbf{W}\epsilon^k}{k^\gamma} \rangle],
\end{align*}}}
where (a) is from $(1-\lambda_n(\textbf{W})+\alpha L)$-Lipschitz condition for $L_{\alpha}(\textbf{x})$; (b) is from Equation \ref{Eq:updating}.
Take the condtional expectation on the above inequality:
{\small{\begin{align*}
\mathbb{E}&[L_\alpha(\textbf{x}^{k+1})|\mathcal{F}_k] \\
&\le L_\alpha(\textbf{x}^{k}) - \|\nabla L_\alpha(\textbf{x}^{k})\|^2 \\
&+ \frac{(1-\lambda_n(\textbf{W})+\alpha L)}{2}[\| \nabla L_{\alpha}(\textbf{x}^k)\|^2+\mathbb{E}[\|\frac{\textbf{W}\epsilon^k}{k^\gamma}\|^2] ]\\
&\stackrel{(c)}{\le} L_\alpha(\textbf{x}^{k}) - \|\nabla L_\alpha(\textbf{x}^{k})\|^2\\
&+ \frac{(1-\lambda_n(\textbf{W})+\alpha L)}{2}[\| \nabla L_{\alpha}(\textbf{x}^k)\|^2+ \frac{n\sigma^2}{k^{2\gamma}} ]\\
& = L_\alpha(\textbf{x}^{k}) - (1-\frac{(1-\lambda_n(\textbf{W})+\alpha L)}{2})\| \nabla L_{\alpha}(\textbf{x}^k)\|^2 \\
&+ \frac{(1-\lambda_n(\textbf{W})+\alpha L)n\sigma^2}{2k^{2\gamma}},
\end{align*}}}
where (c) is from the inequality \ref{Eq:Var}. Thus, with $\alpha < \frac{1+\lambda_n(\textbf{W})}{L}$ and take the full expectation:
{\small{\begin{align*}
&\mathbb{E}[L_\alpha(\textbf{x}^{k+1})] + \frac{(1+\lambda_n(\textbf{W})-\alpha L)}{2}\mathbb{E}[\| \nabla L_{\alpha}(\textbf{x}^k)\|^2] \\
&\le \mathbb{E}[L_\alpha(\textbf{x}^{k})]+ \frac{(1-\lambda_n(\textbf{W})+\alpha L)
n\sigma^2}{2k^{2\gamma}},
\end{align*}}}
where $\frac{(1+\lambda_n(\textbf{W})-\alpha L)}{2}>0.$
Telescoping the above inequality, we have:
{\small{\begin{align}\label{Eq:telescoping}
&\mathbb{E}[L_\alpha(\textbf{x}^{k+1})] + \sum_{i=0}^{k} \frac{(1+\lambda_n(\textbf{W})-\alpha L)}{2}\mathbb{E}[\| \nabla L_{\alpha}(\textbf{x}^i)\|^2] \notag\\
& \le L_\alpha(\textbf{x}^{0})+ \sum_{i=0}^{k} \frac{(1-\lambda_n(\textbf{W})+\alpha L)n\sigma^2}{2i^{2\gamma}}, 
\end{align}}}
which implies that 
{\small{\begin{align}\label{Eq:BddEquality}
&\sum_{i=0}^{k} \frac{(1+\lambda_n(\textbf{W})-\alpha L)}{2}\mathbb{E}[\| \nabla L_{\alpha}(\textbf{x}^i)\|^2] \\
&\le  L_\alpha(\textbf{x}^{0})- \mathbb{E}[L_\alpha(\textbf{x}^{k+1})] + \sum_{i=0}^{k} \frac{(1-\lambda_n(\textbf{W})+\alpha L)n\sigma^2}{2i^{2\gamma}} \notag \\
& \stackrel{(d)}{\le}  L_\alpha(\textbf{x}^{0})-  \alpha\sum_{i=1}^n f_i(x^*) + \sum_{i=0}^{k} \frac{(1-\lambda_n(\textbf{W})+\alpha L)n\sigma^2}{2i^{2\gamma}} 
\end{align}}}
where (d) is because 1) the eigenvalues of $(\textbf{I}-\textbf{W})$ are positive and $\frac{1}{2}\textbf{x}^T(\textbf{I}-\textbf{W})\textbf{x}\ge 0$ and 2) $\sum_{i=1}^{n}f_i(x^*)\le \sum_{i=1}^{n}f_i(x),$ $\forall x,$ thus, 
{\small{\begin{equation}
L_\alpha(\textbf{x}^{k+1})= \frac{1}{2}(\textbf{x}^{k+1})^T(\textbf{I}-\textbf{W})\textbf{x}^{k+1} + \alpha\sum_{i=1}^{n}f_i(x_{i,k+1})\ge \alpha\sum_{i=1}^{n}f_i(x^*).\notag
\end{equation} }}
With the condition $\gamma >\frac{1}{2},$ Equation (\ref{Eq:BddEquality}) indicates that $\sum_{i=0}^{\infty}\mathbb{E}[\| \nabla L_{\alpha}(\textbf{x}^i)\|^2]$ is bounded and $\mathbb{E}[\| \nabla L_{\alpha}(\textbf{x}^i)\|^2] = o(1/k),$ due to Lemma \ref{lemma:series}. Also, from Equation \ref{Eq:telescoping}, $\mathbb{E}[L_\alpha(\textbf{x}^{k})]<+\infty$ holds, $\forall k$.

Due to the fact that
\begin{equation}
\mathbb{E}[\sum_{i=1}^{n}f_i(x_{i,k})]\le \mathbb{E}[\frac{1}{\alpha}L_\alpha(\textbf{x}^{k})]<+\infty,\notag
\end{equation} with Assumption \ref{Assump:ObjInfity}, the following holds:
{\small{
\begin{align*}
\mathbb{E}[\|\textbf{x}^k&\|]\le\int_{\|\textbf{x}^k\|\le A} \|\textbf{x}^k\|p(\textbf{x}^k)d(\textbf{x}^k)+\int_{\|\textbf{x}^k\|> A} \|\textbf{x}^k\|p(\textbf{x}^k)d(\textbf{x}^k)\\
&\le \int_{\|\textbf{x}^k\|\le A} A p(\textbf{x}^k)d(\textbf{x}^k)+\int_{\|\textbf{x}^k\|> A} \frac{\|\textbf{x}^k\|}{f(\textbf{x}^k)}f(\textbf{x}^k)p(\textbf{x}^k)d(\textbf{x}^k)\\
&\le A+(M+\delta_A)\mathbb{E}[f(\textbf{x}^k)] \le +\infty
\end{align*}}}
where $p(\textbf{x}^k)$ is the density of $\textbf{x}^k$ and $\delta_A$ is from the fact that $\forall A>0$, there exist $\delta_A<\infty,$ such that $|\frac{\|\textbf{x}^k\|}{f(\textbf{x}^k)}|<M+\delta_A.$ Hence $\mathbb{E}[\|\textbf{x}^k\|]$ is bounded, $\forall k,$ say bounded by $B.$ Then,
{\small{
\begin{align*}
\mathbb{E}[\|\nabla f(\textbf{x}^k)\|]&=\mathbb{E}[\|\nabla f(\textbf{x}^k)-\nabla f(\textbf{x}^*) \|] \\
&\le L\mathbb{E}[\|\textbf{x}^k-\textbf{x}^*\|]\le L(B+\|\textbf{x}^*\|).
\end{align*}}}
\end{proof}

\section{Proof for Lemma\ref{Lemma:convolution}}
First of all, $h_k$ could be bounded by a integral $h_k = \beta^k \sum_{i=1}^{k}\frac{\beta^{-i}}{i^\gamma}\le \beta^k(\int_{x=1}^{k}\frac{\beta^{-x}}{x}dx+k)$.
Thus, we focus on the increasing rate of $\int_{x=1}^{k}\frac{\beta^{-x}}{x}dx.$ Denote $z:=-\log(\beta)>0$ and $s:=1-\gamma\in (0,1),$ we have the following complex integral:
\begin{align*}
&\int_{x=1}^{k}\frac{\beta^{-x}}{x^\gamma} \!=\! \int_{x=1}^{k}\frac{e^{zx}}{x^{1-s}}dx \!\stackrel{(a)}{=}\! \int_{y=z}^{zk}\frac{e^{y}}{(y/z)^{1-s}z}dy  \\
&\!=\! z^{-s} \!\! \int_{y=z}^{zk}\frac{e^{y}}{y^{1-s}}dy \stackrel{(b)}{=}z^{-s}\int_{t=e^{-i\pi}z}^{e^{-i\pi}zk}e^{-t}t^{s-1}e^{i\pi s}dt\\
&=z^{-s}e^{i\pi s}[\int_{t=e^{-i\pi}z}^{\infty}e^{-t}t^{s-1}dt-\int_{t=e^{-i\pi}kz}^{\infty}e^{-t}t^{s-1}dt]\\
&\stackrel{(c)}{=}z^{-s}e^{i\pi s}[\Gamma(s,e^{-i\pi}z)-\Gamma(s,e^{-i\pi}kz)],
\end{align*}
where $y=zx$ in (a), $t=e^{-i\pi}y$ in (b) with $e^{-i\pi}+1=0$ and $\Gamma(\cdot,\cdot)$ is the upper incomplete gamma function in (c). With the properity of $\Gamma(\cdot,\cdot),$ it follows that as $k\rightarrow +\infty,$
$\Gamma(s,e^{-i\pi}kz)=(e^{-i\pi}kz)^{s-1}e^{kz}[1+O(\frac{1}{k})]$, which implies the second term on the above equation is $z^{-s}e^{i\pi s}\Gamma(s,e^{-i\pi}kz)=-z^{-1}k^{-\gamma}\beta^{-k}[1+O(\frac{1}{k})]$.
Hence, the integral is $\int_{x=1}^{k}\frac{\beta^{k-x}}{x^\gamma}=\beta^k\{z^{-s}e^{i\pi s}\Gamma(s,e^{-i\pi}z)+z^{-1}k^{-\gamma} \beta^{-k}[1+O(\frac{1}{k})]\} =\frac{1}{k^{-\gamma}z}[1+O(\frac{1}{k})]+O(\beta^k)$,
and $h_k=O(\frac{1}{k^\gamma}).$
This completes the proof.

\section{Proof for Theorem~\ref{thm:conv_meanvec}} \label{appdx:conv_meanvec}
\begin{proof}
It is knowns that that $\bar{\textbf{x}}^k= \frac{1}{n}\textbf{11}^T \textbf{x}^k.$
As a result,
{\small{\begin{align}
&\|\textbf{x}^k-\bar{\textbf{x}}^k\| =\|\textbf{x}^k-\frac{1}{n}\textbf{11}^T\textbf{x}^k\| \notag\\
&\stackrel{(a)}{=}\|(\textbf{I}-\frac{1}{n}\textbf{11}^T)\sum_{i=0}^{k-1}[- \alpha_i\textbf{W}^{k-i-1}\nabla f(\textbf{x}^i) +  \textbf{W}^{k-i}\frac{\epsilon^i}{i^\gamma}] \| \notag\\
& \le   \sum_{i=1}^{k-1}[\alpha_i\|\textbf{W}^{k-i-1}-\frac{1}{n}\textbf{11}^T\|\|\nabla f(\textbf{x}^i)\| +\|\textbf{W}^{k-i}-\frac{1}{n}\textbf{11}^T\|\|\frac{\epsilon^i}{i^{\gamma}}\|]\notag\\
&\stackrel{(b)}{\le}\sum_{i=1}^{k-1}\alpha_i\beta^{k-i-1}\|\nabla f(\textbf{x}^i)\|+\sum_{i=1}^{k-1}\beta^{k-i}\|\frac{\epsilon^i}{i^{\gamma}}\|
\end{align}}}
where (a) is from (\ref{Eq:recursive}) and (b) is because the largest eigenvalue of $(\textbf{W}-\frac{1}{n}\textbf{11}^T)$ is bounded by $\beta.$
Now take the full expectation on the both sides:
{\small{\begin{equation}
\mathbb{E}[\|\textbf{x}^k-\frac{1}{n}\textbf{11}^T \textbf{x}^k\|] \le  \sum_{i=1}^{k-1}\alpha_i\beta^{k-i-1}\mathbb{E}[\|\nabla f(\textbf{x}^i)\|]+\sum_{i=1}^{k-1}\beta^{k-i}\frac{\sqrt{n}\sigma}{i^{\gamma}}.
\end{equation}}}
Hence, with Lemma \ref{Lemma:convolution}, if $\mathbb{E}[\|\nabla f(\textbf{x}^i)\|]$ is bounded by $D$ and step-size $\alpha_i=\alpha$, $\forall i,$ then it follows that $\mathbb{E}[\|\textbf{x}^k-\frac{1}{n}\textbf{11}^T \textbf{x}^k\|] \le \frac{\alpha D}{1-\beta}+O(\frac{\sqrt{n}\sigma}{k^\gamma});$ while the step-size $\alpha_i=O(\frac{1}{k^\eta})$ with $\eta>0,$ then $\mathbb{E}[\|\textbf{x}^k-\frac{1}{n}\textbf{11}^T \textbf{x}^k\|] \le \sum_{i=1}^{k-1}\beta^{k-i}(\frac{D}{\beta}\alpha_i+\sqrt{n}\sigma\frac{1}{k^\gamma}) = O(\frac{1}{k^{\min(\eta,\gamma)}}).$
\end{proof}

\section{Proof for Lemma \ref{Lemma:5}}
\begin{proof}
With $L$-Lipschitz condition for the local objectives, we have:
{\small{\begin{align}\label{Eq:T3_1}
&\frac{1}{n}\sum_{i=1}^{n} f_i(\bar{x}^{k+1}) \notag\\
&\le \frac{1}{n}\sum_{i=1}^{n} f_i(\bar{x}^{k})+\langle \frac{1}{n}\sum_{i=1}^{n} \nabla f_i(\bar{x}^{k}), \bar{x}^{k+1}-\bar{x}^{k} \rangle +\frac{L}{2}\|\bar{x}^{k+1}-\bar{x}^{k}\|^2 \notag\\
&\stackrel{(a)}{=}  \frac{1}{n}\sum_{i=1}^{n} f_i(\bar{x}^{k}) + \langle \frac{1}{n}\sum_{i=1}^{n} \nabla f_i(\bar{x}^{k}), -\frac{\alpha_k}{n}\sum_{i=1}^{n} \nabla f_i(x_{i,k})+\notag\\
&\frac{1}{n}\sum_{i=1}^{n} \frac{\epsilon_{k^\gamma y_{i,k}}}{k^\gamma} \rangle +\frac{L}{2}\|-\frac{\alpha_k}{n}\sum_{i=1}^{n} \nabla f_i(x_{i,k})+\frac{1}{n}\sum_{i=1}^{n} \frac{\epsilon_{k^\gamma y_{i,k}}}{k^\gamma}\|^2 
\end{align}}}
where (a) is becaue of the fact:
{\small{\begin{align*}
\bar{x}^{k+1}&=\frac{1}{n}\sum_{i,j=1}^{n}\textbf{W}_{ij}x_{j,k+1}-\frac{\alpha_k}{n}\sum_{i=1}^{n}\nabla f_{i}(x_{i,k})+\frac{1}{n}\sum_{i,j=1}^{n}\textbf{W}_{ij}\frac{\epsilon_{k^\gamma y_{j,k}}}{k^\gamma} \\
&= \bar{x}^k- \frac{\alpha_k}{n}\sum_{i=1}^{n}\nabla f_{i}(x_{i,k}) + \frac{1}{n} \sum_{i=1}^{n}\frac{\epsilon_{k^\gamma y_{i,k}}}{k^\gamma}.
\end{align*}}}
Take the conditional expectation on Equation (\ref{Eq:T3_1}) and simply calculate:
{\small{\begin{align*}\label{Eq:T3_2}
&\mathbb{E}[\frac{1}{n}\sum_{i=1}^{n} f_i(\bar{x}^{k+1})|\mathcal{F}_k]\notag \\
& \le  \frac{1}{n}\sum_{i=1}^{n} f_i(\bar{x}^{k}) + \langle \frac{1}{n}\sum_{i=1}^{n} \nabla f_i(\bar{x}^{k}), -\frac{\alpha_k}{n}\sum_{i=1}^{n} \nabla f_i(x_{i,k})+\mathbb{E}[\frac{1}{n}\sum_{i=1}^{n} \\
&\frac{\epsilon_{k^\gamma y_{i,k}}}{k^\gamma}|\mathcal{F}_k] \rangle +\frac{L}{2}\mathbb{E}[\|-\frac{\alpha_k}{n}\sum_{i=1}^{n} \nabla f_i(x_{i,k})+\frac{1}{n}\sum_{i=1}^{n} \frac{\epsilon_{k^\gamma y_{i,k}}}{k^\gamma}\|^2 |\mathcal{F}_k]  \\
& \stackrel{(b)}{=} \frac{1}{n}\sum_{i=1}^{n} f_i(\bar{x}^{k}) + \langle \frac{1}{n}\sum_{i=1}^{n} \nabla f_i(\bar{x}^{k}), -\frac{\alpha_k}{n}\sum_{i=1}^{n} \nabla f_i(x_{i,k}) \rangle  + \frac{L}{2}\{\alpha_k^2\\
&\|\frac{1}{n}\sum_{i=1}^{n} \nabla f_i(x_{i,k})\|^2+\mathbb{E}[\|\frac{1}{n}\sum_{i=1}^{n} \frac{\epsilon_{k^\gamma y_{i,k}}}{k^\gamma}\|^2 |\mathcal{F}_k]\} \\
& \stackrel{(c)}{\le}\frac{1}{n}\sum_{i=1}^{n} f_i(\bar{x}^{k}) -\alpha_k \langle \frac{1}{n}\sum_{i=1}^{n} \nabla f_i(\bar{x}^{k}), \frac{1}{n}\sum_{i=1}^{n} \nabla f_i(x_{i,k}) \rangle  + \frac{L}{2}(\|\frac{1}{n}\\
&\sum_{i=1}^{n} \nabla f_i(x_{i,k})\|^2+\frac{p\sigma^2}{nk^{2\gamma}}) \\
& \stackrel{(d)}{=} \frac{1}{n}\sum_{i=1}^{n} f_i(\bar{x}^{k})- \alpha_k \|\frac{1}{n}\sum_{i=1}^{n} \nabla f_i(\bar{x}^{k})\|^2  + \alpha_k \langle \frac{1}{n}\sum_{i=1}^{n} \nabla f_i(\bar{x}^{k}), \frac{1}{n}\\
&\sum_{i=1}^{n} \nabla f_i(\bar{x}^{k})-\frac{1}{n}\sum_{i=1}^{n} \nabla f_i(x_{i,k})  \rangle \\
& + \frac{L\alpha_k^2}{2}(\|\frac{1}{n}\sum_{i=1}^{n} \nabla f_i(\bar{x}^{k})\|^2+\|\frac{1}{n}\sum_{i=1}^{n} \nabla f_i(\bar{x}^{k})-\frac{1}{n}\sum_{i=1}^{n} \nabla f_i(x_{i,k})\|^2\\
&-2\langle \frac{1}{n}\sum_{i=1}^{n} \nabla f_i(\bar{x}^{k}), \frac{1}{n}\sum_{i=1}^{n} \nabla f_i(\bar{x}^{k})-\frac{1}{n}\sum_{i=1}^{n} \nabla f_i(x_{i,k}) \rangle) + \frac{L\sigma^2}{2nk^{2\gamma}} \\
& = \frac{1}{n}\sum_{i=1}^{n} f_i(\bar{x}^{k}) - (\alpha_k-\frac{L\alpha_k^2}{2})\|\frac{1}{n}\sum_{i=1}^{n} \nabla f_i(\bar{x}^{k})\|^2 +\frac{L\alpha_k^2}{2}\|\frac{1}{n}\sum_{i=1}^{n} \displaybreak[3]\\
&\nabla f_i(\bar{x}^{k})-\frac{1}{n}\sum_{i=1}^{n} \nabla f_i(x_{i,k})\|^2 + (\alpha_k-L\alpha_k^2) \langle \frac{1}{n}\sum_{i=1}^{n} \nabla f_i(\bar{x}^{k}), \frac{1}{n}\\
&\sum_{i=1}^{n} \nabla f_i(\bar{x}^{k})-\frac{1}{n}\sum_{i=1}^{n} \nabla f_i(x_{i,k})  \rangle + \frac{L\sigma^2}{2nk^{2\gamma}}\\
& \stackrel{(e)}{\le} \frac{1}{n}\sum_{i=1}^{n} f_i(\bar{x}^{k}) - (\alpha_k-\frac{L\alpha_k^2}{2}-\frac{\alpha_k-L\alpha_k^2}{2})\|\frac{1}{n}\sum_{i=1}^{n} \nabla f_i(\bar{x}^{k})\|^2  + \\
&(\frac{L\alpha_k^2}{2}+\frac{\alpha_k-L\alpha_k^2}{2})\|\frac{1}{n}\sum_{i=1}^{n} \nabla f_i(\bar{x}^{k})-\frac{1}{n}\sum_{i=1}^{n} \nabla f_i(x_{i,k})\|^2 + \frac{L\sigma^2}{2nk^{2\gamma}} \\
& = \frac{1}{n}\sum_{i=1}^{n} f_i(\bar{x}^{k}) - \frac{\alpha_k}{2}\|\frac{1}{n}\sum_{i=1}^{n} \nabla f_i(\bar{x}^{k})\|^2  + \frac{\alpha_k}{2}\|\frac{1}{n}\sum_{i=1}^{n} \nabla f_i(\bar{x}^{k})-\\
&\frac{1}{n}\sum_{i=1}^{n} \nabla f_i(x_{i,k})\|^2 + \frac{L\sigma^2}{2nk^{2\gamma}}\\
& \stackrel{(f)}{\le} \frac{1}{n}\sum_{i=1}^{n} f_i(\bar{x}^{k}) - \frac{\alpha_k}{2}\|\frac{1}{n}\sum_{i=1}^{n} \nabla f_i(\bar{x}^{k})\|^2  + \frac{\alpha_k L^2}{2n^2}\sum_{i=1}^{n}\| \bar{x}^{k}- x_{i,k}\|^2 \\
&+ \frac{L\sigma^2}{2nk^{2\gamma}} 
\end{align*}}}

where (b) is because the expectation of $\epsilon_{k^\gamma y_{i,k}}$ is zero, (c) is from Definition \ref{defn:usco}, (d) is by adding and substracting $\alpha_k \|\frac{1}{n}\sum_{i=1}^{n} \nabla f_i(\bar{x}^{k})\|^2$ and decompositing the quadratic term, (e) is from $\pm2\langle a, b \rangle\le \|a\|^2+\|b\|^2$ and (f) is from $L$-Lipschitz condition and Jensen Inequality.

\end{proof}

\section{Proof for Theorem \ref{thm:const_step}} \label{appdx:thm_const_step}
\begin{proof}
From Theorem \ref{thm:bnd_grad}, we know that $\mathbb{E}\|\nabla f(\textbf{x}^k)\| \le L(B+\|\textbf{x}^*\|).$
With the inequality of Lemma \ref{Lemma:5} and Theorem \ref{thm:conv_meanvec}, it holds that:
{\small{\begin{align}
\frac{\alpha_k}{2}\|\frac{1}{n}\sum_{i=1}^{n} \nabla f_i(\bar{x}^{k})\|^2 &\le \frac{1}{n}\sum_{i=1}^{n} f_i(\bar{x}^{k}) - \mathbb{E}[\frac{1}{n}\sum_{i=1}^{n} f_i(\bar{x}^{k+1})|\mathcal{F}_k] \notag\\
&+\frac{\alpha_k L}{2n}(\frac{\alpha_k D}{1-\beta}+O(\frac{\sqrt{n}\sigma}{k^\gamma}))^2 + \frac{Lp\sigma^2}{2nk^{2\gamma}} \notag\\
&= \frac{1}{n}\sum_{i=1}^{n} f_i(\bar{x}^{k}) - \mathbb{E}[\frac{1}{n}\sum_{i=1}^{n} f_i(\bar{x}^{k+1})|\mathcal{F}_k] \notag\\
&+ \frac{\alpha_k^3 L^2D^2}{2n(1-\beta)^2}+O(\frac{\alpha_k L\sqrt{n}\sigma}{nk^\gamma}) \notag\\
& = \frac{1}{n}\sum_{i=1}^{n} f_i(\bar{x}^{k}) - \mathbb{E}[\frac{1}{n}\sum_{i=1}^{n} f_i(\bar{x}^{k+1})|\mathcal{F}_k]\notag \\
&+ \frac{\alpha_k^3 L^4(B+\|\textbf{x}^*\|)^2}{2n(1-\beta)^2}+O(\frac{\alpha_k L\sqrt{n}\sigma}{nk^\gamma}). \notag
\end{align}}}

Take the full expectation and telescope the inequalities:
{\small{\begin{align*}
&\sum_{i=0}^{k}\frac{\alpha}{2}\mathbb{E}[\|\frac{1}{n}\sum_{i=1}^{n}] \nabla f_i(\bar{x}^{k})\|^2 \\
& \le \frac{1}{n}\sum_{i=1}^{n} f_i(0) - \mathbb{E}[\frac{1}{n}\sum_{i=1}^{n} f_i(\bar{x}^{k+1})|\mathcal{F}_k] + \frac{k\alpha^3 L^4(B+\|\textbf{x}^*\|)^2}{2n(1-\beta)^2}\\
&+kO(\frac{\alpha L\sigma}{\sqrt{n}k^\gamma})\notag \\
& \le \frac{1}{n}\sum_{i=1}^{n} f_i(0) - \frac{1}{n}\sum_{i=1}^{n} f_i(x^*) + \frac{k\alpha^3 L^4(B+\|\textbf{x}^*\|)^2}{2n(1-\beta)^2}+kO(\frac{\alpha L\sigma}{\sqrt{n}k^\gamma})\notag
\end{align*}}}
Hence, 
{\small{\begin{align*}
\min_{k} \{\mathbb{E}[\|\frac{1}{n}\sum_{i=1}^{n}] \nabla f_i(\bar{x}^{k})\|^2\} &\le \frac{2}{\alpha k}[\frac{1}{n}\sum_{i=1}^{n} f_i(0) - \frac{1}{n}\sum_{i=1}^{n} f_i(x^*)]
\\
&+ \frac{\alpha^2 L^4(B+\|\textbf{x}^*\|)^2}{n(1-\beta)^2}+O(\frac{2L\sigma}{\sqrt{n}k^\gamma}).\notag
\end{align*}}}
\end{proof}

\section{Proof for Theorem~\ref{thm:diminish_step}} \label{appdx:thm_diminish_step}
\begin{proof}
From Lemma \ref{Lemma:5} and Theorem \ref{thm:conv_meanvec}, it holds that:
{\small{\begin{align*}
&\mathbb{E}[\frac{1}{n}\sum_{i=1}^{n} f_i(\bar{x}^{k+1})|\mathcal{F}_k]+ \frac{\alpha_k}{2}\|\frac{1}{n}\sum_{i=1}^{n} \nabla f_i(\bar{x}^{k})\|^2 \\
&\le \frac{1}{n}\sum_{i=1}^{n} f_i(\bar{x}^{k})  + O(\frac{\alpha_k L^2}{2n^2k^{\min(\eta,\gamma)}}) + \frac{L\sigma^2}{2nk^{2\gamma}}.\notag
\end{align*}}}

Then, we will apply Supermartingale Convergence Theorem to complete the proof. This theorem has been used to proof the convergence of Stochastic Gradient Descent Algorithm, e.g. \cite{hannah2016unbounded,zhang2018taming}. 
\begin{thm}(Supermartingale Convergence Theorem) \label{Theorem:supermartingale}
Let $a^k$, $\theta^k$ and $\rho^k$ be positive sequences adapted to $\mathcal{F}_k$, and let $\rho^k$ be summable with probably $1$. If it holds that 
\begin{equation}
\mathbb{E}[a^{k+1}|\mathcal{F}_k]+\theta^k\le a^k+\rho^k, \notag
\end{equation}
then with probability $1$, $a^k$ converages to a $[0,\infty)$-valued random variable, and $\sum_{k=1}^{\infty} \theta^k<+\infty$.
\end{thm}
In our case, denote $a^k=\frac{1}{n}\sum_{i=1}^{n} f_i(\bar{x}^{k})-\frac{1}{n}\sum_{i=1}^{n} f_i(x^{*}),$ $\theta^k=\frac{\alpha_k}{2}\|\frac{1}{n}\sum_{i=1}^{n} \nabla f_i(\bar{x}^{k})\|^2,$ and $\rho^k=O(\frac{\alpha_k L^2}{2n^2k^{\min(\eta,\gamma)}}) + \frac{Lp\sigma^2}{2nk^{2\gamma}}=O(\frac{1}{k^{\eta+\min(\eta,\gamma)}}+\frac{1}{k^{2\gamma}}).$ With $\gamma>\frac{1}{2}$ and $\eta\ge\frac{1}{2},$ there exists $\epsilon>0$ and $\rho^k=O(\frac{1}{k^{(1+\epsilon)}}).$ With Lemma \ref{lemma:series}, $\rho^k$ is summable with probably $1.$ Then $\theta^k$ is summable almost surely, i.e. $\sum_{i=1}^{\infty}\frac{\alpha_k}{2}\|\frac{1}{n}\sum_{i=1}^{n} \nabla f_i(\bar{x}^{k})\|^2 <+\infty.$ Again, with Lemma \ref{lemma:series}, $\|\frac{1}{n}\sum_{i=1}^{n} \nabla f_i(\bar{x}^{k})\|^2=o(\frac{1}{k^{1-\eta}})$  almost surely.

\end{proof}

\end{document}